\documentclass[runningheads]{llncs}

\usepackage[T1]{fontenc}
\usepackage{xspace}
\usepackage{stmaryrd, amsmath, amssymb, mathtools, mathrsfs, extarrows}
\usepackage{nicefrac}
\usepackage{multirow}
\usepackage{xcolor}
\usepackage{changepage}
\usepackage{wrapfig}
\usepackage{paralist}
\usepackage{tikz}
	\usetikzlibrary{trees, decorations, arrows, automata, positioning, plotmarks, shapes, backgrounds, decorations.pathmorphing, arrows.meta}

\usepackage{macros_poc}

\begin{document}
\title{Probabilistic Operational Correspondence (Technical Report)}

\author{Anna Schmitt
	\inst{1}
	\orcidID{0000-0001-6675-2879}
	\and
	Kirstin Peters
	\inst{2}
	\orcidID{0000-0002-4281-0074}}

\authorrunning{A.\ Schmitt and K.\ Peters}

\institute{TU Darmstadt, Germany \and Augsburg University, Germany}

\maketitle

\begin{adjustwidth}{-2.4cm}{-2.2cm}
\begin{abstract}
	Encodings are the main way to compare process calculi.
	By applying quality criteria to encodings we analyse their quality and rule out trivial or meaningless encodings.
	Thereby, operational correspondence is one of the most common and most important quality criteria.
	It ensures that processes and their translations have the same abstract behaviour.
	We analyse probabilistic versions of operational correspondence to enable such a verification for probabilistic systems.

	Concretely, we present three versions of probabilistic operational correspondence: weak, middle, and strong.
	We show the relevance of the weaker version using an encoding from a sublanguage of probabilistic \CCS into the probabilistic $\pi$-calculus.
	Moreover, we map this version of probabilistic operational correspondence onto a probabilistic behavioural relation that directly relates source and target terms. Then we can analyse the quality of the criterion by analysing the relation it induces between a source term and its translation.
	For the second version of probabilistic operational correspondence we proceed in the opposite direction. We start with a standard simulation relation for probabilistic systems and map it onto a probabilistic operational correspondence criterion.
	
	This technical report contains the proofs to the lemmata and theorems of \cite{schmitt23} as well as some additional material.
\end{abstract}


\section{Process Calculi and Encodings}
\label{app:calculi}

\begin{definition}[Reductions of Distributions]
	\label{def:stepDistributions}
	Let $ \Delta \step \Theta $ whenever
	\begin{enumerate}[(a)]
		\item $ \Delta = \sum_{i \in I} p_i \pointDis{P_i} $, where $ I $ is a finite index set and $ \sum_{i \in I} p_i = 1 $,
		\item for each $ i \in I $ there is a distribution $ \Theta_i $ such that $ P_i \step \Theta_i $ or $ \Theta_i = \pointDis{P_i} $,
		\item for some $ i \in I $ we have $ P_i \step \Theta_i $, and
		\item $ \Theta = \sum_{i \in I} p_i \cdot \Theta_i $.
	\end{enumerate}
\end{definition}

\begin{definition}[Relations on Distributions, \cite{Deng07}]
	\label{def:relationsDistributions}
	$ $\\
	Let $ \relation \subseteq \processes{}^2 $ and let $ \Delta, \Theta \in \mathcal{D}{\left( \processes{} \right)} $.
	Then $ \left( \Delta, \Theta \right) \in \pointDis{\relation} $ if
	\begin{enumerate}[(a)]
		\item $ \Delta = \sum_{i \in I} p_i \pointDis{P_i} $, where $ I $ is a finite index set and $ \sum_{i \in I} p_i = 1 $,
		\item for each $ i \in I $ there is a process $ Q_i $ such that $ \left( P_i, Q_i \right) \in \relation $, and
		\item $ \Theta = \sum_{i \in I} p_i \pointDis{Q_i} $.
	\end{enumerate}
\end{definition}

For our proofs it is important that Definition~\ref{def:relationsDistributions} translates preorders into preorders.
Accordingly, we prove that it preserves reflexivity.

\begin{lemma}[Preservation of Reflexivity]
	\label{lem:reflexivityRelationDistiribution}
	If $ \relation $ is reflexive, then so is $ \pointDis{\relation} $.
\end{lemma}

\begin{proof}
	Assume a reflexive relation $ \relation $ and consider a probability distribution $ \Delta $ with $ \Delta = \sum_{i \in I} p_i \pointDis{P_i} $ for some finite index set $ I $ with $ \sum_{i \in I} p_i = 1 $.
	We have to prove that $ \left( \Delta, \Delta \right) \in \pointDis{\relation} $, \ie that for each $ i \in I $ there is some $ Q_i $ such that $ \left( P_i, Q_i \right) \in \relation $ and $ \Delta = \sum_{i \in I} p_i \pointDis{Q_i} $.
	Since $ \relation $ is reflexive, we have $ \left( P_i, P_i \right) \in \relation $, \ie it suffices to choose $ Q_i = P_i $ to conclude the proof.
\end{proof}

The preservation of transitivity was already given in \cite{Deng07}.

\begin{lemma}[Preservation of Transitivity, \cite{Deng07}]
	\label{lem:transitvityRelationDistiribution}
	$ $\\
	If $ \relation $ is transitive, then so is $ \pointDis{\relation} $.
\end{lemma}

We inherit the criteria (expect operational correspondence) from \cite{Gorla10}:
\begin{description}
	\item[Compositionality:] For every operator ${\mathbf{op}}$ with arity $n$ of $\lang{\source}$ and for every subset of names $N$, there exists a context $\Context{N}{\mathbf{op}}{\hole_1, \ldots, \hole_n}$ such that, for all $S_1, \ldots, S_n$ with $\freeNames{S_1} \cup \ldots \cup \freeNames{S_n} = N$, it holds that $\enc{{\mathbf{op}}\left( S_1, \ldots, S_n \right)} = \Context{N}{\mathbf{op}}{\enc{S_1}, \ldots, \enc{S_n}}$.
	\item[Name Invariance \wrt a Relation $ \relationT \subseteq \processes{\target}^2 $:] For every $ S \in \processes{\source} $ and every substitution $ \sigma $, it holds that $ \enc{S\sigma} \equiv_{\alpha} \enc{S}\sigma' $ if $ \sigma $ is injective and $ \left( \enc{S\sigma}, \enc{S}\sigma' \right) \in \relationT $ otherwise, where $ \sigma' $ is such that $ \renamingFun{\sigma(a)} = \sigma'{\left( \renamingFun{a} \right)} $ for all $ a \in \names $.
	\item[Divergence Reflection:] For every $S$, $\enc{S} \infiniteSteps$ implies $S \infiniteSteps$.
	\item[Success Sensitiveness:] For every $S$, $\ReachBarb{S}{\success}$ iff $\ReachBarb{\enc{S}}{\success}$.
\end{description}

The formulation of compositionality is rather strict, \ie it rules out practically relevant translations.
Note that the best known encoding from the asynchronous $ \pi $-calculus into the Join Calculus in \cite{FournetGonthier96} is not compositional, but consists of an inner, compositional encoding surrounded by a fixed context---the implementation of so-called firewalls---that is parameterised on the free names of the source term. In order to capture this and similar encodings we relax the definition of compositionality.
\begin{description}
	\item[Weak Compositionality:] The encoding is either compositional or consists of an inner, compositional encoding surrounded by a fixed context that can be parameterised on the free names of the source term or information that are not part of the source term.
\end{description}

\begin{definition}[Operational Correspondence, Non-Probabilistic]
	\label{def:operationalCorrespondence}
	$ $\\
	An encoding $ \enc{\cdot} $ is \emph{strongly operationally corresponding} \wrt $ \relationT \subseteq \processes{\target}^2 $ if it is:
	\begin{description}
		\item[\quad Strongly Complete:] $ \forall S, S' \logdot S \step S' \text{ implies } \left( \exists T \logdot \enc{S} \step T \wedge \left( \enc{S'}, T \right) \in \relationT \right) $
		\item[\quad Strongly Sound:] $ \forall S, T \logdot \enc{S} \step T \text{ implies } \left( \exists S' \logdot S \step S' \wedge \left( \enc{S'}, T \right) \in \relationT \right) $
	\end{description}
	$ \enc{\cdot} $ is \emph{operationally corresponding} \wrt $ \relationT \subseteq \processes{\target}^2 $ if it is:
	\begin{description}
		\item[\quad Complete:] $ \forall S, S' \logdot S \transstep S' \text{ implies } \left( \exists T \logdot \enc{S} \transstep T \wedge \left( \enc{S'}, T \right) \in \relationT \right) $
		\item[\quad Sound:] $ \forall S, T \logdot \enc{S} \transstep T \text{ implies } \left( \exists S' \logdot S \transstep S' \wedge \left( \enc{S'}, T \right) \in \relationT \right) $
	\end{description}
	$ \enc{\cdot} $ is \emph{weakly operationally corresponding} \wrt $ \relationT \subseteq \processes{\target}^2 $ if it is:
	\begin{description}
		\item[\quad Complete:] $ \forall S, S' \logdot S \transstep S' \text{ implies } \left( \exists T \logdot \enc{S} \transstep T \wedge \left( \enc{S'}, T \right) \in \relationT \right) $
		\item[\quad Weakly Sound:] $ \forall S, T \logdot \enc{S} \transstep T \text{ impl. } \left( \exists S', T' \logdot S \transstep S' \wedge T \transstep T' \wedge \left( \enc{S'}, T' \right) \in \relationT \right) $
	\end{description}
\end{definition}

\subsection{Probabilistic \CCS}

Probabilistic \CCS is introduced in \cite{Deng07} as a probabilistic extension of \CCS \cite{Milner89} to study probabilistic barbed congruence.
We omit the operator for non-de\-ter\-mi\-nis\-tic choice from \cite{Deng07}; not because it is non-deterministic but because the summands of this choice are not necessarily guarded, whereas our target language has only guarded choice.
We will also adapt the semantics of recursion, to ensure the unfolding of recursion requires a step as it is the case in our target language.
We denote the resulting calculus as $ \PCCS $.
Its syntax is given in the following Definition:
\begin{definition}[Syntax of $ \PCCS $]
	\label{def:ProbCCSSyntax}
	The terms $ \processes{C} $ of $ \PCCS $ are given by:
	\begin{align*}
		P & \deffTerms \probCCSChoice{u}{i \in I}{p_i}{P_i} \sep \para{P_1}{P_2} \sep \probCCSRes{P}{A} \sep P{\left[ f \right]} \sep C{\left\langle \tilde{x} \right\rangle}
	\end{align*}
	where $ A \subseteq \names $ and $ f: \names \to \names $ is a renaming function.
\end{definition}

All names in $ A $ are bound in $ P $ by $ \probCCSRes{P}{A} $ and all names in $ \tilde{x} $ are bound in $ P $ by $ C \stackrel{\text{def}}{=} {\left( \tilde{x} \right)}P $.
Names that are not bound are free.
A renaming function can only affect the free names of a term.
Let $ \freeNames{P} $ be the set of free names in $ P $ such that $ \freeNames{Q{\left[ f \right]}} = \set{f(n) \mid n \in \freeNames{Q}} $ for all $ Q \in \processes{C} $.

Following \cite{Deng07} we extend some operations on processes to distributions, because these notions help us to define the semantics of the respective languages.
Let $ \Delta_1, \Delta_2 $ be distributions on processes.
We define the distributions $ \Delta_1 \mid \Delta_2 $ (for parallel composition), $ \probCCSRes{\Delta_1}{A} $ and $ \probPiRes{x}{\Delta_1} $ (for restriction), and $ \Delta_1{\left[ f \right]} $ (for a renaming function $ f $) as:
\begin{align*}
	{\left( \Delta_1 \mid \Delta_2 \right)}(P) &= \begin{cases} \Delta_1{\left( P_1 \right)} \cdot \Delta_2{\left( P_2 \right)} & \text{, if } P = P_1 \mid P_2\\ 0 & \text{otherwise} \end{cases}\\
	{\left( \probCCSRes{\Delta_1}{A} \right)}(P) &= \begin{cases} \Delta_1{\left( P' \right)} & \text{, if } P = \probCCSRes{P'}{A}\\ 0 & \text{, otherwise} \end{cases}\\
	{\left( \probPiRes{x}{\Delta_1} \right)}(P) &= \begin{cases} \Delta_1{\left( P' \right)} & \text{, if } P = \probPiRes{x}{P'}\\ 0 & \text{, otherwise} \end{cases}\\
	{\left( \Delta_1{\left[ f \right]} \right)}(P) &= \begin{cases} \Delta_1{\left( P' \right)} & \text{, if } P = P'{\left[ f \right]}\\ 0 & \text{, otherwise} \end{cases}
\end{align*}

The semantics of $ \PCCS $ is given by the rules in Figure~\ref{fig:ProbCCSSemantics}, where we start with the labelled semantics of \cite{Deng07}, change the Rule~\probCCSRecRule for recursion, and add the Rule~\probCCSReducRule to obtain a reduction semantics.

\begin{figure}[t]
	\centering
	\begin{displaymath}\begin{array}{c}
			\probCCSProbChoiceRule \; \dfrac{\Delta(P) = \sum \set{ p_i \mid i \in I \wedge P_i = P }}{\probCCSChoice{u}{i \in I}{p_i}{P_i} \labelledstep{u} \Delta}
			\hspace{2em}
			\probCCSParLRule \; \dfrac{P_1 \labelledstep{u} \Delta_1}{\para{P_1}{P_2} \labelledstep{u} \para{\Delta_1}{\pointDis{P_2}}}
			\vspace{0.5em}\\
			\probCCSComLRule \; \dfrac{P_1 \labelledstep{a} \Delta_1 \quad P_2 \labelledstep{\out{a}} \Delta_2}{\para{P_1}{P_2} \labelledstep{\tau} \para{\Delta_1}{\Delta_2}} 
			\hspace{2em}
			\probCCSResRule \; \dfrac{P \labelledstep{u} \Delta \quad u \notin A \cup \out{A}}{\probCCSRes{P}{A} \labelledstep{u} \probCCSRes{\Delta}{A}}
			\vspace{0.5em}\\
			\probCCSRelabelRule \; \dfrac{P \labelledstep{v} \Delta \quad f(v) = u}{P{\left[ f \right]} \labelledstep{u} \Delta{\left[ f \right]}}
			\hspace{2em}
			\probCCSRecRule \; \dfrac{C \stackrel{\text{def}}{=} {\left( \tilde{x} \right)}P}{C{\left\langle \tilde{y} \right\rangle} \labelledstep{\tau} \pointDis{P\sub{\tilde{y}}{\tilde{x}}}}
			\vspace{0.5em}\\
			\probCCSReducRule \; \dfrac{P \labelledstep{\tau} \Delta}{P \step \Delta}
	\end{array}\end{displaymath}
	\vspace{-1.5em}
	\caption{Semantics of $ \PCCS $.} \label{fig:ProbCCSSemantics}
\end{figure}

Rule~\probCCSProbChoiceRule reduces a probabilistic choice to a probability distribution over its branches after performing action $u$.
Rule~\probCCSRecRule instead to \cite{Deng07} makes the unfolding of recursion a separate $ \tau $-step.
The remaining rules are standard \CCS rules adapted to probability distributions, where the symmetric versions of the Rules \probCCSParLRule and \probCCSComLRule are omitted.

\subsection{Probabilistic Pi-Calculus}

The probabilistic $\pi$-calculus ($\PPi$) is introduced in \cite{Varacca07}, as a probabilistic version of the $\pi$I-calculus \cite{Sangiorgi96}, where output is endowed with probabilities.

We assume that names in a vector $\seq{y}$ are pairwise distinct.
The names $ \tilde{y}_i $ are bound in $ P_i $ by $ \probPiBranchIn{x}{i \in I}{i}{\tilde{y}}{P_i} $ and $ \probPiSelectOut{x}{i \in I}{i}{\tilde{y}}{P_i} $; $ x $ is bound in $ P $ by $ \probPiRes{x}{P} $; and the names $ \tilde{y} $ are bound in $ P $ by $ \probPiRep{x}{\tilde{y}}{P} $.
Names that are not bound are free.
Let $ \freeNames{P} $ denote the set of free names in $ P $.

Structural congruence $ \equiv $ is defined, similarly to \cite{Milner99}, as the smallest congruence containing $\alpha$-equivalence $ \equiv_{\alpha} $ that is closed under the following rules:
\begin{align*}
	\begin{array}{c}
		\para{P}{\nul} \equiv P \quad\quad \para{P}{Q} \equiv \para{Q}{P} \quad\quad \para{P}{\paraBrack{Q}{R}} \equiv \para{\paraBrack{P}{Q}}{R}
		\quad\quad
		\probPiRes{x}{\nul} \equiv \nul \quad\quad \probPiRes{xy}{P} \equiv \probPiRes{yx}{P}\\
		\probPiRes{x}{\paraBrack{P}{Q}} \equiv \para{P}{\probPiRes{x}{Q}} \quad \text{ if } x \notin \freeNames{P}
	\end{array}
\end{align*}
We lift structural congruence to distributions, \ie $ \Delta_1 \equiv \Delta_2 $ if there is a finite index set $ I $ such that $ \Delta_1 = \sum_{i \in I} p_i \pointDis{P_i} $, $ \Delta_2 = \sum_{i \in I} p_i \pointDis{Q_i} $, and $ P_i \equiv Q_i $ for all $ i \in I $.
We obtain the same by applying Definition~\ref{def:relationsDistributions} on $ \equiv $ but do not want to use the symbol $ \pointDis{\equiv} $.

The semantics of $ \PPi $ is given by the rules in Figure~\ref{fig:ProbPiSemantics}, where we start with the labelled semantics of \cite{Varacca07} and add the Rule~\probPiReducRule to obtain a reduction semantics.

\begin{figure}[t]
	\centering
	\begin{displaymath}\begin{array}{c}
		\probPiBranchRule \; \dfrac{j \in I}{\probPiBranchIn{x}{i \in I}{i}{\tilde{y}}{P_i}\set{\problabelledstep{\probPiSelectInLabel{x}{\tilde{y}}{j}}{1} P_j}}
		\hspace{2em}
		\probPiRepRule \; \probPiRep{x}{\tilde{y}}{P}\set{\problabelledstep{\probPiInLabel{x}{\tilde{y}}}{1} P \mid \probPiRep{x}{\tilde{y}}{P}}
		\vspace{0.5em}\\
		\probPiSelectRule \; \probPiSelectOut{x}{i \in I}{i}{\tilde{y}}{P_i}\set{\problabelledstep{\probPiSelectOutLabel{x}{\tilde{y}}{i}}{p_i} P_i}_{i \in I}
		\hspace{2em}
		\probPiOutRule \; \probPiOut{x}{\tilde{y}}{P}\set{\problabelledstep{\probPiOutLabel{x}{\tilde{y}}}{1} P}
		\vspace{0.5em}\\
		\probPiResRule \; \dfrac{P\set{\problabelledstep{\beta_i}{p_i} P_i}_{i \in I} \quad \subj{\beta_i} \neq x}{\probPiRes{x}{P}\set{\problabelledstep{\beta_i}{p_i} \probPiRes{x}{P_i}}_{i \in I}}
		\hspace{2em}
		\probPiParLRule \; \dfrac{P\set{\problabelledstep{\beta_i}{p_i} P_i}_{i \in I}}{P \mid Q\set{\problabelledstep{\beta_i}{p_i} P_i \mid Q}_{i \in I}}
		\vspace{0.5em}\\
		\probPiComRule \; \dfrac{P\set{\problabelledstep{\alpha_i}{p_i} P_i}_{i \in I} \quad \forall i \in I \logdot Q\set{\problabelledstep{\beta_i}{1} Q_i} \quad \forall i \in I \logdot \obj{\alpha_i} = \tilde{y}_i}{P \mid Q \set{\problabelledstep{\probPiBullet{\alpha_i}{\beta_i}}{p_i} \probPiRes{\tilde{y}_i}{{\left( P_i \mid Q_i \right)}}}_{i \in I}}
		\vspace{0.5em}\\
		\probPiAlphaRule \; \dfrac{P \equiv_{\alpha} P' \quad P\set{\problabelledstep{\beta_i}{p_i} Q_i}_{i \in I}}{P'\set{\problabelledstep{\beta_i}{p_i} Q_i}_{i \in I}}
		\hspace{1.5em}
		\probPiReducRule \; \dfrac{P\set{\problabelledstep{\tau}{p_i} Q_i}_{i \in I} \; \Delta{\left( R \right)} = \sum\set{ p_i \mid Q_i = R }}{P \step \Delta}
	\end{array}\end{displaymath}
	\vspace{-2em}
	\caption{Semantics of $ \PPi $.} \label{fig:ProbPiSemantics}
\end{figure}

For the labelled part of the semantics we use Labels of the following form: $ \probPiSelectInLabel{x}{\tilde{y}}{i} $, $ \probPiSelectOutLabel{x}{\tilde{y}}{i} $, $ \probPiInLabel{x}{\tilde{y}} $, and $ \probPiOutLabel{x}{\tilde{y}} $.
Rule~\probPiSelectRule implements the behaviour of a probabilistic selected output, which behaves like one of the processes $P_i$, after sending the corresponding output with probability $p_i$.
On the contrary, each input within the branching input is performed with probability $1$, and the process behaves like $P_j$ after receiving $ \tilde{y}_j $ for some $ j \in I $ as defined in Rule~\probPiBranchRule.
Rule~\probPiComRule describes the interaction of input and output, where the passed names are bound.
Here the partial operation $\probPiBullet{}{}$ on labels is formally defined by: $\probPiBullet{\probPiSelectInLabel{x}{\tilde{y}}{i}}{\probPiSelectOutLabel{x}{\tilde{y}}{i}} = \probPiBullet{\probPiInLabel{x}{\tilde{y}}}{\probPiOutLabel{x}{\tilde{y}}} = \tau$ and undefined in all other cases.
The remaining rules are standard $\pi$-calculus rules extended with probabilities, where the symmetric version of \probPiParLRule is omitted.

In \cite{Varacca07} a type system is introduced to ensure some interesting properties of well-typed terms such as linearity.
Here we are only interested in the untyped version of $ \PPi $ and, thus, omit the type system.


\section{\poc for a Reasonable Encoding}
\label{app:encoding}

\begin{figure}[t]
	\begin{center}
		\begin{tabular}{l l l}
			$\EncPCCSPPi{\probCCSChoice{x}{i \in I}{p_i}{P_i}}$ & $=$ & $\probPiInEmpty{x}{\probPiRes{\encNameI}{{\paraBrack{\probPiSelectOutEmpty{\encNameI}{i \in I}{i}{\EncPCCSPPi{P_i}}}{\encNameI}}}}$\\
			$\EncPCCSPPi{\probCCSChoice{\out{x}}{i \in I}{p_i}{P_i}}$ & $=$ & $\probPiSelectOutEmpty{x}{i \in I}{i}{\EncPCCSPPi{P_i}}$\\
			$\EncPCCSPPi{\probCCSChoice{\tau}{i \in I}{p_i}{P_i}}$ & $=$ & $\probPiResBrack{\encNameTau}{\para{\probPiSelectOutEmpty{\encNameTau}{i \in I}{i}{\EncPCCSPPi{P_i}}}{\encNameTau}}$\\
			$\EncPCCSPPi{\para{P}{Q}}$ & $=$ & $\para{\EncPCCSPPi{P}}{\EncPCCSPPi{Q}}$\\
			$\EncPCCSPPi{\probCCSRes{P}{A}}$ & $=$ & $\probPiRes{A}{\EncPCCSPPi{P}}$\\
			$ \EncPCCSPPi{P{\left[ f \right]}} $ & $ = $ & $ \EncPCCSPPi{P}\sub{\mathsf{ran}_f}{\mathsf{dom}_f} $\\
			$ \EncPCCSPPi{C{\left\langle \tilde{y} \right\rangle}} $ & $ = $ & $ \out{C}{\left( \tilde{y} \right)} $\\
			$\EncPCCSPPi{\success}$ & $=$ & $\success$
		\end{tabular}
	\end{center}
	where for each $ f $ the $ \mathsf{ran}_f = y_1, \ldots, y_n $ and $ \mathsf{dom}_f = x_1, \ldots, x_n $ are vectors of names such that $ \set{x_1, \ldots, x_n} = \set{x \mid f(x) \neq x } $ and $ f{\left( x_i \right)} = y_i $ for all $ 1 \leq i \leq n $.
	\caption{Inner Encoding.}
	\label{fig:enc}
\end{figure}

\begin{definition}[Encoding $ \outerEncoding $/$ \encPCCSPPi $ from $ \PCCS $ into $ \PPi $]
	\label{def:encCCSPi}
	The encoding of $ S \in \processes{C} $ with the process definitions $ C_1 \stackrel{\text{def}}{=} {\left( \tilde{x}_1 \right)}.S_1, \ldots, C_n \stackrel{\text{def}}{=} {\left( \tilde{x}_n \right)}.S_n $ consists of the outer encoding $ \outerEncoding $, where $ \OuterEncoding{S} $ is
	\begin{align*}
		\probPiRes{C_1, \ldots, C_n}{\left( \EncPCCSPPi{S} \mid \probPiRep{C_1}{\tilde{x}_1}{\EncPCCSPPi{S_1}} \mid \ldots \mid \probPiRep{C_n}{\tilde{x}_n}{\EncPCCSPPi{S_n}} \right)}
	\end{align*}
	and the inner encoding $ \encPCCSPPi $ is given in Figure~\ref{fig:enc}.
\end{definition}

In Definition~\ref{def:encCCSPi} the encoding $ \outerEncoding $ from $ \PCCS $ into $ \PPi $ is presented.
In the following we prove that this encoding satisfies the criteria given in Section~\ref{app:calculi} (except for a classical version of operational correspondence) and the new criterion weak \poc.

The encoding of a probabilistic choice is split into three cases: the first three cases of Definition~\ref{def:encCCSPi}.
For input guards a single input on $ x $ is used, to enable the communication with a potential corresponding output.
In the following such a communication step on a source term name is denoted as \textit{\firststep-step}.
In the continuation of the input on $ x $ a probabilistic selecting output on the reserved name $\encNameI$ composed in parallel with a matching input is used to encode the probabilities.
This step on the reserved channel name $\encNameI$ is denoted as \secondstep-step.
The sequence of these two communication steps on $x$ and $\encNameI$ emulates the behaviour of a single communication step in the source.

By restricting the scope of the reserved name $\encNameI$, interactions with other operators communicating on $\encNameI$ between two translations of inputs are prevented.
Further, as the renaming policy ensures that $\encNameI$ does not appear in $ \OuterEncoding{P_i} $, conflicts between the reserved name $\encNameI$ and source term names are avoided.

\begin{definition}[\firststep-step]
	\label{def:firststep}
	An \firststep-step is a communication step on a translated source term name.
\end{definition}

\begin{definition}[\secondstep-step]
	\label{def:secondstep}
	An \secondstep-step is a communication step on an instance of the reserved name $\encNameI$.
\end{definition}

The encoding of an output-guarded probabilistic choice is straight forward, as it is translated using the probabilistic selecting output.

For the guard $\tau$, an output-guarded probabilistic choice in parallel to a single input on the reserved name $\encNameTau$ is used.
Because of the restriction, interactions with other translations of $\tau$-guarded operators are prevented.
A communication step of this kind is denoted as \taustep-step.
This step does not only introduce the probabilities of a $ \tau $-guarded source term choice in the translation but also allows the translated term to do a step whenever the source term does one and compensates the missing $ \tau $ in the syntax of the target language.

\begin{definition}[\taustep-step]
	\label{def:taustep}
	An \taustep-step is a communication step on an instance of the reserved name $\encNameTau$.
\end{definition}

The application of a renaming function is encoded by a substitution.
A call $ C{\left\langle \tilde{y} \right\rangle} $ is encoded by an output, where the corresponding process definitions are translated into replicated inputs and placed in parallel by the outer encoding.
The remaining translations are homomorphic.

\begin{definition}[\repstep-step]
	\label{def:repstep}
	An \repstep-step is a communication step that reduces a replicated input.
\end{definition}

\begin{example}[\taustep-Steps]
	\label{exa:tausequence}
	Consider the source term $S = \para{\tau.{\left(\frac{1}{8}P \oplus \frac{7}{8}Q\right)}}{\tau.{\left(\frac{3}{5}R \oplus \frac{2}{5}S\right)}}$ of $ \PCCS $ without process definitions.
	$S$ can do the following sequence of steps:
	\begin{align*}
		S &\step \Delta_{S, 1} = \set{\frac{1}{8} {\left( \para{P}{\tau.{\left(\frac{3}{5}R \oplus \frac{2}{5}S\right)}} \right)}, \frac{7}{8} {\left( \para{Q}{\tau.{\left(\frac{3}{5}R \oplus \frac{2}{5}S\right)}} \right)}}\\
		&\step \Delta_{S, 2} = \set{\frac{3}{40} {\left( P \mid R \right)}, \frac{2}{40} {\left( P \mid S \right)}, \frac{21}{40} {\left( Q \mid R \right)}, \frac{14}{40} {\left( Q \mid S \right)}}
	\end{align*}
	By Definition~\ref{def:encCCSPi} and since $ S $ has no process definitions, $ \OuterEncoding{S} = \EncPCCSPPi{S} $ and:
	\begin{align*}
		\EncPCCSPPi{S} ={}& \probPiResBrack{\encNameTau}{{\para{\out{\encNameTau}\left(\cont{\frac{1}{8}\branchIn_1}{\EncPCCSPPi{P}} \oplus \cont{\frac{7}{8}\branchIn_2}{\EncPCCSPPi{Q}}\right)}{\encNameTau}}}\mid\\
		& \probPiResBrack{\encNameTau}{\para{\out{\encNameTau}\left(\cont{\frac{3}{5}\branchIn_1}{\EncPCCSPPi{R}} \oplus \cont{\frac{2}{5}\branchIn_2}{\EncPCCSPPi{S}}\right)}{\encNameTau}}.
	\end{align*}
	Thereby, the restriction of the reserved name $\encNameTau$ prevents a communication between the left and right subterm of the outermost parallel operator.
	By Figure~\ref{fig:ProbPiSemantics}, $ \OuterEncoding{S} $ can emulate the steps of $ S $ by $ \OuterEncoding{S} \step \Delta_{T, 1} \step \Delta_{T, 2} $, where:
	\begin{align*}
		\Delta_{T, 1} & = \Big\{\begin{array}[t]{l}
				\frac{1}{8} \para{\paraBrack{\EncPCCSPPi{P}}{\nul}}{\probPiResBrack{\encNameTau}{\para{\out{\encNameTau}\left(\cont{\frac{3}{5}\branchIn_1}{\EncPCCSPPi{R}} \oplus \cont{\frac{2}{5}\branchIn_2}{\EncPCCSPPi{S}}\right)}{\encNameTau}}}\\
				\frac{7}{8} \para{\paraBrack{\EncPCCSPPi{Q}}{\nul}}{\probPiResBrack{\encNameTau}{\para{\out{\encNameTau}\left(\cont{\frac{3}{5}\branchIn_1}{\EncPCCSPPi{R}} \oplus \cont{\frac{2}{5}\branchIn_2}{\EncPCCSPPi{S}}\right)}{\encNameTau}}} \Big\}
			\end{array}\\
		\Delta_{T, 2} &= \big\{ \begin{array}[t]{l}
				\frac{3}{40} {\left( \EncPCCSPPi{P} \mid \nul \mid \EncPCCSPPi{R} \mid \nul \right)}, \frac{2}{40} {\left( \EncPCCSPPi{P} \mid \nul \mid \EncPCCSPPi{S} \mid \nul \right)},\\
				\frac{21}{40} {\left( \EncPCCSPPi{Q} \mid \nul \mid \EncPCCSPPi{R} \mid \nul \right)}, \frac{14}{40} {\left( \EncPCCSPPi{Q} \mid \nul \mid \EncPCCSPPi{S} \mid \nul \right)} \big\}
			\end{array}
	\end{align*}
	The distributions $ \Delta_{T, 1} $ and $ \Delta_{T, 2} $ are both structural congruent to the encoding of the corresponding source distributions, \ie
	$ \OuterEncoding{\Delta_{S, 1}} \equiv \Delta_{T, 1} $ and $ \OuterEncoding{\Delta_{S, 2}} \equiv \Delta_{T, 2} $.
	As both steps in $ \OuterEncoding{S} \step \Delta_{T, 1} \step \Delta_{T, 2} $ reduce an instance of $ \encNameTau $---though of course different instances of $ \encNameTau $ are reduced---they are both \taustep-steps. \qed
\end{example}

An example of a \firststep-step followed by a \secondstep-step is presented in \cite{schmitt23}.
To obtain the probability distribution that results from this sequence of two steps on the target, the probabilities of the \secondstep-step are multiplied with the probabilities of the corresponding \firststep-step.
The resulting probabilities match to the probabilities of the emulated source term step.
This multiplication stems from the rules of probability theory, where the probability of an event consisting of a sequence of several events has to be calculated by multiplying the probabilities of all the single events contained in that sequence.
Accordingly, if a single source term step is emulated by a sequence of target term steps, we compare the probabilities that result from multiplying the probabilities of the target term steps in the sequence with the probabilities of the source.
Fortunately, the multiplication of probabilities is already covered by Definition~\ref{def:stepDistributions} in order to define sequences of steps.
It remains to ensure that our version of operational correspondence compares the probabilities of the distribution that results from a single source term step with the final distribution in the emulating target term sequence (and not with the probabilities of a distribution in the middle of this sequence).

We already ruled out strong operational correspondence as defined in Definition~\ref{def:operationalCorrespondence}.
The other two versions differ in whether they allow for intermediate states.
Another look at the example in \cite{schmitt23} tells us, that intermediate states make sense.
$ \Delta_T $ is a finite probability distribution with the probabilities $ \frac{3}{4} $ and $ \frac{1}{4} $, but neither $ S $ nor $ S' $ have cases with these probabilities.
However, since there is exactly one source term step and exactly one sequence of target term steps, $ \Delta_T $ does not mark a partial commitment, because there was nothing to decide.
Indeed the restriction on $ \encNameI $ ensures that each \firststep-step enables exactly one \secondstep-step and communication is the only case that requires two steps to emulate a single source term steps.
Hence, also by interleaving with other emulations, we do not obtain partial commitments.
Nonetheless $ \Delta_T $ is an intermediate state; not intermediate in terms of decisions and commitments but intermediate in terms of probabilities.
In the second variant of operational correspondence in Definition~\ref{def:operationalCorrespondence} without intermediate states, we would need to find a relation $ \relationT $ that relates $ \Delta_T $ either to $ \OuterEncoding{S} $ or $ \Delta_T' $.
Such a relation $ \relationT $ is difficult or at least not intuitive, since it has to relate states with different probabilities.
It is easier to allow for intermediate states.
So, we want to build a weak version of operational correspondence (third case of Definition~\ref{def:operationalCorrespondence}) with probabilities.
This leads to the version of probabilistic operational correspondence below denoted as \textit{weak probabilistic operational correspondence}.

\begin{definition}[Weak Probabilistic Operational Correspondence]
	\label{def:weakPOC}
	An encoding $\enc{\cdot} : \processes{\source} \to \processes{\target}$ is \emph{weakly probabilistic operationally corresponding} (weak \poc) \wrt $\relationT \subseteq \processes{\target}^2 $ if it is:
	\begin{description}
		\item[\quad Probabilistic Complete:] $ $\\
		\hspace*{2em} $ \forall S, \Delta_S \logdot S \transstep \Delta_S \text{ implies } \left( \exists \Delta_T \logdot \enc{S} \transstep \Delta_T \wedge \left( \enc{\Delta_S}, \Delta_T \right) \in \pointDis{\relationT} \right) $
		\item[\quad Weakly Probabilistic Sound:] $ \forall S, \Delta_T \logdot \enc{S} \transstep \Delta_T \text{ implies } $\\
		\hspace*{2em} $ \left( \exists \Delta_S, \Delta_T' \logdot S \transstep \Delta_S \wedge \Delta_T \transstep \Delta_T' \wedge \left( \enc{\Delta_S}, \Delta_T' \right) \in \pointDis{\relationT} \right) $
	\end{description}
\end{definition}

Before we analyse the quality of our new version of \poc in Section~\ref{app:weakPOC}, we want to check whether it indeed exactly captures the way our encoding $ \outerEncoding $ treads source term steps into probability distributions, \ie we prove that our encoding satisfies weak \poc.
The Example~\ref{exa:tausequence} and the example given in \cite{schmitt23} illustrate that steps on $ \tau $-guarded choices and communication steps satisfy weak \poc \wrt $ \equiv $.
They cover \taustep-steps, \firststep-steps, and \secondstep-steps.
The only missing kind of steps, are steps to unfold a recursion in the source and their emulation by \repstep-steps in the target.

\begin{example}[\repstep-Step]
	\label{exa:recursion}
	Consider $ S = C{\left\langle \tilde{y} \right\rangle} $ with $ C \stackrel{\text{def}}{=} {\left( \right)}.\success $ in $ \PCCS $.
	By Figure~\ref{fig:ProbCCSSemantics}, $ S $ can perform only one step: $ S \step \Delta_S = \pointDis{\success} = \set{ 1 \success } $.
	By Definition~\ref{def:encCCSPi}, then:
	\begin{align*}
		\OuterEncoding{S} &= \probPiRes{C}{\left( \EncPCCSPPi{S} \mid \probPiRep{C}{}{\success} \right)}
		\quad \text{ and } \quad
		\EncPCCSPPi{S} = \out{C}{\left( \tilde{y} \right)}
	\end{align*}
	By Figure~\ref{fig:ProbPiSemantics}, $ \OuterEncoding{S} $ can perform exactly one maximal sequence of steps, namely the \repstep-step $ \OuterEncoding{S} \step \Delta_T = \set{ 1 \probPiRes{C}{\left( \success \mid \probPiRep{C}{}{\success} \right)} } $.
	By Definition~\ref{def:encCCSPi}, $ \OuterEncoding{\Delta_S} = \Delta_T $, because even though $ \Delta_S $ does no longer contain any process constants its process definition is not consumed in the step $ S \step \Delta_S $. \qed
\end{example}

We prove that the encoding in Definition~\ref{def:encCCSPi} satisfies the quality criteria of Section~\ref{app:calculi} and weak \poc.
We start with weak compositionality.

\begin{lemma}[Weak Compositionality, $ \outerEncoding $/$ \encPCCSPPi $]
	\label{lem:weakCompositionality}
	$ $\\
	The encoding $ \outerEncoding $ is weakly compositional.
\end{lemma}

\begin{proof}
	Our encoding consists of the outer encoding $ \outerEncoding $ and the inner encoding $ \encPCCSPPi $.
	The outer encoding $ \outerEncoding $ is a fixed context that is parametrised on the process definitions of the source term, that are not part of the source term itself.
	The inner encoding $ \encPCCSPPi $ is compositional, because the encoding function in Definition~\ref{def:encCCSPi} defines a context for each operator of the source language in that the translations of the subterms of the respective source term are used.
	Hence, $ \outerEncoding $ is weakly compositional.
\end{proof}

Name invariance and the different versions of operational correspondence are defined modulo a relation $ \relationT $ on target terms that is success sensitive.
For our encoding $ \outerEncoding $ we can choose $ \relationT $ as the structural congruence $ \equiv $ on the target language $ \PPi $.
Structural congruence satisfies a stronger version of success sensitiveness with $ \HasBarb{\cdot}{\success} $ instead of $ \ReachBarb{\cdot}{\success} $.

\begin{lemma}[$ \equiv $ is Success Sensitive]
	\label{lem:successSensitivenessTargetRelation}
	If $ T_1 \equiv T_2 $ then $ \HasBarb{T_1}{\success} \longleftrightarrow \HasBarb{T_2}{\success} $.\\
	Moreover, if $ \Delta_1 \equiv \Delta_2 $ then $ \HasBarb{\Delta_1}{\success} \longleftrightarrow \HasBarb{\Delta_2}{\success} $.
\end{lemma}

\begin{proof}
	The proof is by induction on the definition of $ \equiv $.
	All cases are immediate.
	\begin{description}
		\item[$ \alpha $-Equivalence $ \equiv_{\alpha} $:] In this case $ T_1 \equiv_{\alpha} T_2 $.
			Since $ \success $ does not contain any names, we have $ \HasBarb{T_1}{\success} $ iff $ \HasBarb{T_2}{\success} $.
		\item[$ \para{P}{\nul} \equiv P $:] In this case $ T_1 = T_2 \mid \nul $.
			Since $ \nul $ does not contain $ \success $, then $ \HasBarb{T_1}{\success} $ iff $ \HasBarb{T_2}{\success} $.
		\item[$ \para{P}{Q} \equiv \para{Q}{P} $:] In this case $ T_1 = P \mid Q $ and $ T_2 = Q \mid P $.
			Since $ T_1 $ contains $ \success $ iff $ T_2 $ contains $ \success $, then $ \HasBarb{T_1}{\success} $ iff $ \HasBarb{T_2}{\success} $.
		\item[$ \para{P}{\paraBrack{Q}{R}} \equiv \para{\paraBrack{P}{Q}}{R} $:] In this case $ T_1 = P \mid {\left( Q \mid R \right)} $ and $ T_2 = {\left( P \mid Q \right)} \mid R $.
			Since $ T_1 $ contains $ \success $ iff $ T_2 $ contains $ \success $, then $ \HasBarb{T_1}{\success} $ iff $ \HasBarb{T_2}{\success} $.
		\item[$ \probPiRes{x}{\nul} \equiv \nul $:] In this case $ T_1 = \probPiRes{x}{\nul} $ and $ T_2 = \nul $.
			Since $ \nul $ does not contain $ \success $, then $ \NotHasBarb{T_1}{\success} $ and $ \NotHasBarb{T_2}{\success} $.
		\item[$ \probPiRes{xy}{P} \equiv \probPiRes{yx}{P} $:] Then $ T_1 = \probPiRes{xy}{P} $ and $ T_2 = \probPiRes{yx}{P} $.
			Since $ T_1 $ as well as $ T_2 $ contain $ \success $ iff $ P $ contains $ \success $, then $ \HasBarb{T_1}{\success} $ iff $ \HasBarb{T_2}{\success} $.
		\item[$ \probPiRes{x}{\paraBrack{P}{Q}} \equiv \para{P}{\probPiRes{x}{Q}} $:] In this case $ T_1 = \probPiRes{x}{\left( P \mid Q \right)} $ and $ T_2 = P \mid \probPiRes{x}{Q} $, where $ x \notin \freeNames{P} $.
			Since $ T_1 $ contains $ \success $ iff $ T_2 $ contains $ \success $, then $ \HasBarb{T_1}{\success} $ iff $ \HasBarb{T_2}{\success} $.
		\item[$ \Delta_1 \equiv \Delta_2 $:] In this case there is a finite index set $ I $ such that $ \Delta_1 = \sum_{i \in I} p_i \pointDis{P_i} $, $ \Delta_2 = \sum_{i \in I} p_i \pointDis{Q_i} $, and $ P_i \equiv Q_i $ for all $ i \in I $.
			Because of $ P_i \equiv Q_i $, we have $ \HasBarb{P_i}{\success} $ iff $ \HasBarb{Q_i}{\success} $ for all $ i \in I $.
			Then $ \HasBarb{\Delta_1}{\success} $ iff $ \HasBarb{\Delta_2}{\success} $. \qed
	\end{description}
\end{proof}

The renaming policy $ \renamingPCCSPPi $ of $ \outerEncoding $ reserves the names $ \encNameI, \encNameTau $ and keeps process constants $ C $ distinct from source term names, \ie $ \length{\RenamingPCCSPPi{n}} = 1 $ and $ \RenamingPCCSPPi{n} \cap \set{\left( \encNameI \right), \left( \encNameTau \right), \left( C \right) \mid C \text{ is a process constant}} = \emptyset $ for all $ n \in \names $.

Name invariance ensures that the encoding function treads all source term names in the same way.
Since the encoding function $ \outerEncoding $ does not introduce any free names and because of the rigorous use of the renaming policy $ \renamingPCCSPPi $ (although we omit it for readability in Definition~\ref{def:encCCSPi}), our encoding satisfies a stronger version of name invariance, where $ \alpha $-equivalence can be used regardless of whether $ \sigma $ is injective.

\begin{lemma}[Name Invariance, $ \outerEncoding $/$ \encPCCSPPi $]
	\label{lem:nameInvariance}
	For every $ S \in \processes{\source} $ and every substitution $ \sigma $, it holds that $ \OuterEncoding{S\sigma} \equiv_{\alpha} \OuterEncoding{S}\sigma' $ and $ \EncPCCSPPi{S\sigma} \equiv_{\alpha} \EncPCCSPPi{S}\sigma' $, where $ \sigma' $ is such that $ \RenamingPCCSPPi{\sigma(a)} = \sigma'{\left( \RenamingPCCSPPi{a} \right)} $ for all $ a \in \names $.
\end{lemma}

\begin{proof}
	Without loss of generality we assume that $ \sigma' $ behaves as identity for all names that are not in the range of $ \renamingPCCSPPi $.
	The assumption can be replaced by applying alpha conversion such that the names introduced by the encoding function, \ie the restricted names that are denoted by $ \encNameI $, $ \encNameTau $, or $ C_i $ in Definition~\ref{def:encCCSPi}, are not affected by applying $ \sigma' $.
	The proof is by induction on the encoding function.
	\begin{description}
		\item[$ \OuterEncoding{S} $:] Assume without loss of generality that $ S $ has the process definitions $ C_1 \stackrel{\text{def}}{=} {\left( \tilde{x}_1 \right)}.S_1, \ldots, C_n \stackrel{\text{def}}{=} {\left( \tilde{x}_n \right)}.S_n $, where $ \freeNames{S_i} = \set{\tilde{x}_i} $ for all $ 1 \leq i \leq n $.
			By the induction hypothesis, $ \EncPCCSPPi{S\sigma} \equiv_{\alpha} \EncPCCSPPi{S}\sigma' $.
			The renaming policy ensures that $ \RenamingPCCSPPi{n} \cap \set{C_1, \ldots, C_n} = \emptyset $ for all $ n \in \names $.
			Then we have:
			\begin{align*}
				& \OuterEncoding{S\sigma}\\
				& = \probPiRes{C_1, \ldots, C_n}{\left( \EncPCCSPPi{S\sigma} \mid \probPiRep{C_1}{\tilde{x}_1}{\EncPCCSPPi{S_1}} \mid \ldots \mid \probPiRep{C_n}{\tilde{x}_n}{\EncPCCSPPi{S_n}} \right)}\\
				& \equiv_{\alpha} \probPiRes{C_1, \ldots, C_n}{\left( \EncPCCSPPi{S}\sigma' \mid \probPiRep{C_1}{\tilde{x}_1}{\EncPCCSPPi{S_1}} \mid \ldots \mid \probPiRep{C_n}{\tilde{x}_n}{\EncPCCSPPi{S_n}} \right)}\\
				& \equiv_{\alpha} \probPiRes{C_1, \ldots, C_n}{\left( \EncPCCSPPi{S} \mid \probPiRep{C_1}{\tilde{x}_1}{\EncPCCSPPi{S_1}} \mid \ldots \mid \probPiRep{C_n}{\tilde{x}_n}{\EncPCCSPPi{S_n}} \right)}\sigma'\\
				& = \OuterEncoding{S}\sigma'
			\end{align*}
			Note that here $ \tilde{x}_i $ is short for the sequence that results from applying $ \renamingPCCSPPi $ on all names in $ \tilde{x}_i $.
			Because of that, the renaming policy $ \renamingPCCSPPi $ ensures that $ \freeNames{\probPiRep{C_i}{\tilde{x}_i}{\EncPCCSPPi{S_i}}} = \set{C_i} $ and thus that $ \sigma' $ has no effect on these terms.
		\item[$ \EncPCCSPPi{\probCCSChoice{x}{i \in I}{p_i}{P_i}} $:] In this case $ S = \probCCSChoice{x}{i \in I}{p_i}{P_i} $.
			By the induction hypothesis, $ \EncPCCSPPi{P\sigma} \equiv_{\alpha} \EncPCCSPPi{P}\sigma' $.
			The renaming policy ensures that $ \RenamingPCCSPPi{n} \cap \set{\encNameI} = \emptyset $ for all $ n \in \names $.
			Then we have:
			\begin{align*}
				\EncPCCSPPi{S\sigma} &= \EncPCCSPPi{\probCCSChoice{\sigma(x)}{i \in I}{p_i}{P_i\sigma}}\\
				& = \probPiInEmpty{\sigma'(x)}{\probPiRes{\encNameI}{{\paraBrack{\probPiSelectOutEmpty{\encNameI}{i \in I}{i}{\EncPCCSPPi{P_i\sigma}}}{\encNameI}}}}\\
				& \equiv_{\alpha} \probPiInEmpty{\sigma'(x)}{\probPiRes{\encNameI}{{\paraBrack{\probPiSelectOutEmpty{\encNameI}{i \in I}{i}{\EncPCCSPPi{P_i}\sigma'}}{\encNameI}}}}\\
				& \equiv_{\alpha} \probPiInEmpty{x}{\probPiRes{\encNameI}{{\paraBrack{\probPiSelectOutEmpty{\encNameI}{i \in I}{i}{\EncPCCSPPi{P_i}}}{\encNameI}}}}\sigma' = \EncPCCSPPi{S}\sigma'
			\end{align*}
		\item[$ \EncPCCSPPi{\probCCSChoice{\out{x}}{i \in I}{p_i}{P_i}} $:] In this case $ S = \probCCSChoice{\out{x}}{i \in I}{p_i}{P_i} $.
			By the induction hypothesis, $ \EncPCCSPPi{P\sigma} \equiv_{\alpha} \EncPCCSPPi{P}\sigma' $.
			Then we have:
			\begin{align*}
				\EncPCCSPPi{S\sigma} &= \EncPCCSPPi{\probCCSChoice{\out{\sigma(x)}}{i \in I}{p_i}{P_i\sigma}}
				= \probPiSelectOutEmpty{\sigma'(x)}{i \in I}{i}{\EncPCCSPPi{P_i\sigma}}\\
				& \equiv_{\alpha} \probPiSelectOutEmpty{\sigma'(x)}{i \in I}{i}{\EncPCCSPPi{P_i}\sigma'}\\
				& \equiv_{\alpha} {\left( \probPiSelectOutEmpty{x}{i \in I}{i}{\EncPCCSPPi{P_i}} \right)}\sigma' = \EncPCCSPPi{S}\sigma'
			\end{align*}
		\item[$ \EncPCCSPPi{\probCCSChoice{\tau}{i \in I}{p_i}{P_i}} $:] In this case $ S = \probCCSChoice{\tau}{i \in I}{p_i}{P_i} $.
			By the induction hypothesis, $ \EncPCCSPPi{P\sigma} \equiv_{\alpha} \EncPCCSPPi{P}\sigma' $.
			The renaming policy ensures that $ \RenamingPCCSPPi{n} \cap \set{\encNameTau} = \emptyset $ for all $ n \in \names $.
			Then we have:
			\begin{align*}
				\EncPCCSPPi{S\sigma} & = \EncPCCSPPi{\probCCSChoice{\tau}{i \in I}{p_i}{P_i\sigma}}
				\equiv_{\alpha} \probPiResBrack{\encNameTau}{\para{\probPiSelectOutEmpty{\encNameTau}{i \in I}{i}{\EncPCCSPPi{P_i\sigma}}}{\encNameTau}}\\
				& \equiv_{\alpha} \probPiResBrack{\encNameTau}{\para{\probPiSelectOutEmpty{\encNameTau}{i \in I}{i}{\EncPCCSPPi{P_i}\sigma'}}{\encNameTau}}\\
				& \equiv_{\alpha} \probPiResBrack{\encNameTau}{\para{\probPiSelectOutEmpty{\encNameTau}{i \in I}{i}{\EncPCCSPPi{P_i}}}{\encNameTau}}\sigma'
				= \EncPCCSPPi{S}\sigma'
			\end{align*}
		\item[$ \EncPCCSPPi{\para{P}{Q}} $:] In this case $ S = P \mid Q $.
			By the induction hypothesis, $ \EncPCCSPPi{P\sigma} \equiv_{\alpha} \EncPCCSPPi{P}\sigma' $ and $ \EncPCCSPPi{Q\sigma} \equiv_{\alpha} \EncPCCSPPi{Q}\sigma' $.
			Thereby, $ \EncPCCSPPi{S\sigma} = \EncPCCSPPi{P\sigma \mid Q\sigma} = \EncPCCSPPi{P\sigma} \mid \EncPCCSPPi{Q\sigma} \equiv_{\alpha} \EncPCCSPPi{P}\sigma' \mid \EncPCCSPPi{Q}\sigma' = {\left( \EncPCCSPPi{P} \mid \EncPCCSPPi{Q} \right)}\sigma' = \EncPCCSPPi{S}\sigma' $.
		\item[$ \EncPCCSPPi{\probCCSRes{P}{A}} $:] In this case $ S = \probCCSRes{P}{A} $.
			Let $ \gamma $ be obtained from $ \sigma $ by removing all names in $ A $ from the domain of $ \sigma $.
			Moreover, let $ \gamma' $ be such that $ \RenamingPCCSPPi{\gamma(a)} = \gamma'{\left( \RenamingPCCSPPi{a} \right)} $ for all $ a \in \names $.
			By the induction hypothesis, $ \EncPCCSPPi{P\gamma} \equiv_{\alpha} \EncPCCSPPi{P}\gamma' $.
			Then $ \EncPCCSPPi{S\sigma} = \EncPCCSPPi{\probCCSRes{P\gamma}{A}} = \probPiRes{A}{\EncPCCSPPi{P\gamma}} \equiv_{\alpha} \probPiRes{A}{{\left( \EncPCCSPPi{P}\gamma' \right)}} = {\left( \probPiRes{A}{\EncPCCSPPi{P}} \right)}\sigma' = \EncPCCSPPi{S}\sigma' $.
		\item[$ \EncPCCSPPi{P{\left[ f \right]}} $:] In this case $ S = P{\left[ f \right]} $.
			Let $ f' $ be such that $ f'(\sigma(n)) = \sigma(f(n)) $ for all $ n \in \names $.
			By the induction hypothesis, $ \EncPCCSPPi{P\sigma} \equiv_{\alpha} \EncPCCSPPi{P}\sigma' $.
			Then we have:
			\begin{align*}
				\EncPCCSPPi{S\sigma} & = \EncPCCSPPi{{\left( P\sigma \right)}{\left[ f' \right]}} = \EncPCCSPPi{P\sigma}\sub{\mathsf{ran}_{f'}}{\mathsf{dom}_{f'}}\\
				& \equiv_{\alpha} {\left( \EncPCCSPPi{P}\sigma' \right)}\sub{\mathsf{ran}_{f'}}{\mathsf{dom}_{f'}}\\
				& = {\left( \EncPCCSPPi{P}\sub{\mathsf{ran}_f}{\mathsf{dom}_f} \right)}\sigma' = \EncPCCSPPi{S}\sigma'
			\end{align*}
		\item[$ \EncPCCSPPi{C{\left\langle \tilde{y} \right\rangle}} $:] In this case $ S = C{\left\langle \tilde{y} \right\rangle} $.
			Let $ \tilde{z} $ be the result of applying $ \sigma $ on all names in $ \tilde{y} $.
			Then $ \EncPCCSPPi{S\sigma} = \EncPCCSPPi{C{\left\langle \tilde{z} \right\rangle}} = \out{C}{\left( \tilde{z} \right)} = \out{C}{\left( \tilde{y} \right)}\sigma' = \EncPCCSPPi{S}\sigma' $.
		\item[$ \EncPCCSPPi{\success} $:] In this case $ S = \success $.
			Then $ \EncPCCSPPi{S\sigma} = \EncPCCSPPi{\success} = \success = \success\sigma' = \EncPCCSPPi{S}\sigma' $. \qed
	\end{description}
\end{proof}

We introduced in Definition~\ref{def:weakPOC} a new variant of operational correspondence, namely weak probabilistic operational correspondence (weak \poc), for the encoding $ \outerEncoding $.
For the completeness part, we have to prove that the encoding preserves the behaviour of source terms.
Therefore, we show how the translations emulate a source term step.

\begin{lemma}[Weak \poc, Completeness, $ \outerEncoding $/$ \encPCCSPPi $]
	\label{lem:completeness}
	\begin{align*}
		\forall S, \Delta_S \logdot S \transstep \Delta_S \text{ implies } \left( \exists \Delta_T \logdot \OuterEncoding{S} \transstep \Delta_T \wedge \OuterEncoding{\Delta_S} \equiv \Delta_T \right)
	\end{align*}
\end{lemma}

\begin{proof}
	We start with a single step $ S \step \Delta_S $ and show that we need in this case a finite and non-empty sequence of steps $ \OuterEncoding{S} \transstep \Delta_T $ in the target such that $ \OuterEncoding{\Delta_S} \equiv \Delta_T $.
	Let $ C_1, \ldots, C_n $ be all process constants in $ S $ and $ C_1 \stackrel{\text{def}}{=} {\left( \tilde{x}_1 \right)}.S_1, \ldots, C_n \stackrel{\text{def}}{=} {\left( \tilde{x}_n \right)}.S_n $ the corresponding process definitions.
	Then $ \OuterEncoding{S} = \probPiRes{C_1, \ldots, C_n}{\left( \EncPCCSPPi{S} \mid \probPiRep{C_1}{\tilde{x}_1}{\EncPCCSPPi{S_1}} \mid \ldots \mid \probPiRep{C_n}{\tilde{x}_n}{\EncPCCSPPi{S_n}} \right)} $.\linebreak
	Since the outer restriction on $ C_1, \ldots, C_n $ and the subterms $ \probPiRep{C_i}{\tilde{x}_i}{\EncPCCSPPi{S_i}} $ are not altered by steps of the target term, we define the context
	\begin{align*}
		\Context{}{}{\hole} = \probPiRes{C_1, \ldots, C_n}{\left( \hole \mid \probPiRep{C_1}{\tilde{x}_1}{\EncPCCSPPi{S_1}} \mid \ldots \mid \probPiRep{C_n}{\tilde{x}_n}{\EncPCCSPPi{S_n}} \right)}
	\end{align*}
	to capture this part of target terms, \ie $ \OuterEncoding{S} = \Context{}{}{\EncPCCSPPi{S}} $.
	By Figure~\ref{fig:ProbCCSSemantics}, $ S \step \Delta_S $ was derived from the Rule~\probCCSReducRule, \ie $ S \labelledstep{\tau} \Delta_S $.
	To strengthen our induction hypothesis and to capture labels different from $ \tau $ we prove
	\begin{align*}
		\forall S, \Delta_S \logdot S \labelledstep{u} \Delta_S \text{ implies } \left( \exists T_i \logdot \OuterEncoding{S}\set{ \xLongrightarrow[p_i]{\hat{u}} T_i }_{i \in I} \wedge \OuterEncoding{\Delta_S} \equiv \sum_{i \in I}p_i T_i \right)
	\end{align*}
	where $ \xLongrightarrow[p_i]{\hat{\tau}} $ is $ \xlongrightarrow[p_{i, 1}]{\tau}\cdots\xlongrightarrow[p_{i, n}]{\tau} $, $ \xLongrightarrow[p_i]{\hat{x}} $ is $ \xlongrightarrow[p_{i, 1}]{\tau}\cdots\xlongrightarrow[p_{i, j}]{x\branchIn_i\langle \rangle}\cdots\xlongrightarrow[p_{i, n}]{\tau} $, $ \xLongrightarrow[p_i]{\hat{\out{x}}} $ is $ \xlongrightarrow[p_{i, 1}]{\tau}\cdots\xlongrightarrow[p_{i, j}]{\out{x}\branchIn_i\langle \rangle}\cdots\xlongrightarrow[p_{i, n}]{\tau} $, and in all three cases $ p_i = p_{i, 1} \cdot \ldots \cdot p_{i, n} $.
	We perform an induction over the derivation of $ S \labelledstep{u} \Delta_S $ using a case split over the rules in Figure~\ref{fig:ProbCCSSemantics}.
	\begin{description}
		\item[\probCCSProbChoiceRule:] We consider three subcases:
			\begin{description}
				\item[$ u = x $:] In this case $ S = \probCCSChoice{x}{i \in I}{p_i}{P_i} $ as well as $ \Delta_S = \sum_{i \in I} p_i P_i $.
					By Definition~\ref{def:encCCSPi}, then $ \EncPCCSPPi{S} = \probPiInEmpty{x}{\probPiRes{\encNameI}{{\paraBrack{\probPiSelectOutEmpty{\encNameI}{i \in I}{i}{\EncPCCSPPi{P_i}}}{\encNameI}}}} $ and we have $ \OuterEncoding{\Delta_S} = \sum_{i \in I} p_i \Context{}{}{\EncPCCSPPi{P_i}} $.
					$ \OuterEncoding{S} $ can emulate the step $ S \labelledstep{\tau} \Delta_S $ using the Rules \probPiResRule, \probPiParLRule, \probPiComRule, \probPiBranchRule, and \probPiSelectRule by:
					\begin{align*}
						\OuterEncoding{S}\set{ \xlongrightarrow[1]{x\branchIn_i\langle \rangle}\xlongrightarrow[p_i]{\tau} T_i}_{i \in I} \quad \text{ where } T_i = \Context{}{}{\probPiRes{\encNameI}{{\left( \EncPCCSPPi{P_i} \mid \nul \right)}}}
					\end{align*}
					The renaming policy $ \renamingPCCSPPi $ ensures that $ \encNameI \notin \freeNames{\EncPCCSPPi{P_i}} $ for all $ i \in I $.
					Then $ \OuterEncoding{\Delta_S} \equiv \sum_{i \in I} p_i T_i $.
				\item[$ u = \out{x} $:] In this case $ S = \probCCSChoice{\out{x}}{i \in I}{p_i}{P_i} $ and $ \Delta_S = \sum_{i \in I} p_i P_i $.
					By Definition~\ref{def:encCCSPi}, then $ \EncPCCSPPi{S} = \probPiSelectOutEmpty{x}{i \in I}{i}{\EncPCCSPPi{P_i}} $ and we have $ \OuterEncoding{\Delta_S} = \sum_{i \in I} p_i \Context{}{}{\EncPCCSPPi{P_i}} $.
					$ \OuterEncoding{S} $ can emulate the step $ S \labelledstep{\tau} \Delta_S $ using the Rules \probPiResRule, \probPiParLRule, and \probPiSelectRule by
					\begin{align*}
						\OuterEncoding{S}\set{ \xlongrightarrow[p_i]{\out{x}\branchIn_i\langle \rangle} T_i}_{i \in I} \quad \text{ where } T_i = \Context{}{}{\EncPCCSPPi{P_i}}
					\end{align*}
					and where the Rules \probPiResRule and \probPiParLRule are necessary to do steps in the inner part of the encoding.
					Then $ \OuterEncoding{\Delta_S} = \sum_{i \in I} p_i T_i $ and thus $ \OuterEncoding{\Delta_S} \equiv \sum_{i \in I} p_i T_i $.
				\item[$ u = \tau $:] In this case $ S = \probCCSChoice{\tau}{i \in I}{p_i}{P_i} $ as well as $ \Delta_S = \sum_{i \in I} p_i P_i $.
					By Definition~\ref{def:encCCSPi}, then $ \EncPCCSPPi{S} = \probPiResBrack{\encNameTau}{\para{\probPiSelectOutEmpty{\encNameTau}{i \in I}{i}{\EncPCCSPPi{P_i}}}{\encNameTau}} $ and we have $ \OuterEncoding{\Delta_S} = \sum_{i \in I} p_i \Context{}{}{\EncPCCSPPi{P_i}} $.
					$ \OuterEncoding{S} $ can emulate the step $ S \labelledstep{\tau} \Delta_S $ using the Rules \probPiResRule, \probPiParLRule, \probPiComRule, \probPiBranchRule, and \probPiSelectRule by:
					\begin{align*}
						\OuterEncoding{S}\set{ \xlongrightarrow[p_i]{\tau} T_i}_{i \in I} \quad \text{ where } T_i = \Context{}{}{\probPiRes{\encNameTau}{{\left( \EncPCCSPPi{P_i} \mid \nul \right)}}}
					\end{align*}
					The renaming policy $ \renamingPCCSPPi $ ensures that $ \encNameTau \notin \freeNames{\EncPCCSPPi{P_i}} $ for all $ i \in I $.
					Then $ \OuterEncoding{\Delta_S} \equiv \sum_{i \in I} p_i T_i $.
			\end{description}
		\item[\probCCSParLRule:] In this case $ S = \para{P}{Q} $, $ P \labelledstep{u} \Delta_P = \sum_{i \in I} p_i P_i $, and $ \Delta_S = \Delta_P \mid \pointDis{Q} = \sum_{i \in I} p_i {\left( P_i \mid Q \right)} $.
			By Definition~\ref{def:encCCSPi}, then $ \EncPCCSPPi{S} = \EncPCCSPPi{P} \mid \EncPCCSPPi{Q} $ and $ \OuterEncoding{\Delta_S} = \sum_{i \in I} p_i \Context{}{}{\EncPCCSPPi{P_i} \mid \EncPCCSPPi{Q}} $.
			By the induction hypothesis, the step $ P \labelledstep{u} \Delta_P $ implies $ \OuterEncoding{P}\set{ \xLongrightarrow[p_i]{\hat{u}} T_{i, P} }_{i \in I} $ and $ \OuterEncoding{\Delta_P} \equiv \sum_{i \in I}p_i T_{i, P} $.
			$ \OuterEncoding{S} $ can emulate the step $ S \labelledstep{u} \Delta_S $ using the Rules \probPiResRule and \probPiParLRule to apply the steps in $ \OuterEncoding{P}\set{ \xLongrightarrow[p_i]{\hat{u}} T_{i, P} }_{i \in I} $ such that:
			\begin{align*}
				\OuterEncoding{S}\set{ \xLongrightarrow[p_i]{\hat{u}} T_i }_{i \in I} \quad \text{ where } T_i = \Context{}{}{T_{i, P}' \mid \EncPCCSPPi{Q}} \text{ and } T_{i, P} = \Context{}{}{T_{i, P}'}
			\end{align*}
			Because $ \OuterEncoding{\Delta_P} \equiv \sum_{i \in I}p_i T_{i, P} $ and $ \OuterEncoding{\Delta_P} = \sum_{i \in I}p_i \Context{}{}{\EncPCCSPPi{P_i}} $, then $ \OuterEncoding{\Delta_S} \equiv \sum_{i \in I} p_i T_i $.
		\item[\probCCSParRRule:] In this case $ S = \para{P}{Q} $, $ Q \labelledstep{u} \Delta_Q = \sum_{i \in I} p_i Q_i $, and $ \Delta_S = \pointDis{P} \mid \Delta_Q = \sum_{i \in I} p_i {\left( P \mid Q_i \right)} $.
			By Definition~\ref{def:encCCSPi}, then $ \EncPCCSPPi{S} = \EncPCCSPPi{P} \mid \EncPCCSPPi{Q} $ and $ \OuterEncoding{\Delta_S} = \sum_{i \in I} p_i \Context{}{}{\EncPCCSPPi{P} \mid \EncPCCSPPi{Q_i}} $.
			By the induction hypothesis, the step $ Q \labelledstep{u} \Delta_Q $ implies $ \OuterEncoding{Q}\set{ \xLongrightarrow[p_i]{\hat{u}} T_{i, Q} }_{i \in I} $ and $ \OuterEncoding{\Delta_Q} \equiv \sum_{i \in I}p_i T_{i, Q} $.
			$ \OuterEncoding{S} $ can emulate the step $ S \labelledstep{u} \Delta_S $ using the Rules \probPiResRule, \probPiParLRule, and \probPiParRRule to apply the steps in $ \OuterEncoding{Q}\set{ \xLongrightarrow[p_i]{\hat{u}} T_{i, Q} }_{i \in I} $ such that:
			\begin{align*}
				\OuterEncoding{S}\set{ \xLongrightarrow[p_i]{\hat{u}} T_i }_{i \in I} \quad \text{ where } T_i = \Context{}{}{\EncPCCSPPi{P} \mid T_{i, Q}'} \text{ and } T_{i, Q} = \Context{}{}{T_{i, Q}'}
			\end{align*}
			Because $ \OuterEncoding{\Delta_Q} \equiv \sum_{i \in I}p_i T_{i, Q} $ and $ \OuterEncoding{\Delta_Q} = \sum_{i \in I}p_i \Context{}{}{\EncPCCSPPi{Q_i}} $, then $ \OuterEncoding{\Delta_S} \equiv \sum_{i \in I} p_i T_i $.
		\item[\probCCSComLRule:] Here $ S = P \mid Q $, $ P \labelledstep{a} \Delta_P = \sum_{i \in I} p_i P_i $, $ Q \labelledstep{\out{a}} \Delta_Q = \sum_{j \in J} p_j Q_j $, and we have $ \Delta_S = \Delta_P \mid \Delta_Q = \sum_{i \in I, j \in J} p_i \cdot p_j {\left( P_i \mid Q_j \right)} $.
			By Definition~\ref{def:encCCSPi}, then $ \EncPCCSPPi{S} = \EncPCCSPPi{P} \mid \EncPCCSPPi{Q} $ and we have $ \OuterEncoding{\Delta_S} = \OuterEncoding{\Delta_P} \mid \OuterEncoding{\Delta_Q} = \sum_{i \in I, j \in J} p_i \cdot p_j \Context{}{}{\EncPCCSPPi{P_i} \mid \EncPCCSPPi{Q_j}} $.
			By the induction hypothesis, the step $ P \labelledstep{a} \Delta_P $ implies $ \OuterEncoding{P}\set{ \xLongrightarrow[p_i]{\hat{a}} T_{i, P} }_{i \in I} $ and $ \OuterEncoding{\Delta_P} \equiv \sum_{i \in I}p_i T_{i, P} $.
			By the induction hypothesis, the step $ Q \labelledstep{\out{a}} \Delta_Q $ implies $ \OuterEncoding{Q}\set{ \xLongrightarrow[p_j]{\hat{\out{a}}} T_{j, Q} }_{j \in J} $ and $ \OuterEncoding{\Delta_Q} \equiv \sum_{j \in J}p_j T_{j, Q} $.
			$ \OuterEncoding{S} $ can emulate the step $ S \labelledstep{\tau} \Delta_S $ using the Rules \probPiResRule, \probPiParLRule, and \probPiParRRule to apply the steps in the sequences $ \OuterEncoding{P}\set{ \xLongrightarrow[p_i]{\hat{u}} T_{i, P} }_{i \in I} $ and $ \OuterEncoding{Q}\set{ \xLongrightarrow[p_j]{\hat{u}} T_{j, Q} }_{j \in J} $ such that:
			\begin{align*}
				\OuterEncoding{S}\set{ \xLongrightarrow[p_i \cdot p_j]{\hat{\tau}} T_{i, j} }_{i \in I, j \in J} \quad & \text{where } T_{i, j} = \Context{}{}{T_{i, P}' \mid T_{j, Q}'} \text{, }\\
				& T_{i, P} = \Context{}{}{T_{i, P}'} \text{, and } T_{j, Q} = \Context{}{}{T_{j, Q}'}
			\end{align*}
			Because $ \OuterEncoding{\Delta_P} \equiv \sum_{i \in I}p_i T_{i, P} $ and $ \OuterEncoding{\Delta_P} = \sum_{i \in I}p_i \Context{}{}{\EncPCCSPPi{P_i}} $ and $ \OuterEncoding{\Delta_Q} \equiv \sum_{j \in J}p_j T_{j, Q} $ and $ \OuterEncoding{\Delta_Q} = \sum_{j \in J}p_j \Context{}{}{\EncPCCSPPi{Q_j}} $, then $ \OuterEncoding{\Delta_S} \equiv \sum_{i \in I, j \in J} p_i \cdot p_j T_{i, j} $.
		\item[\probCCSComRRule:] This case is symmetric to the last case for \probCCSComLRule.
		\item[\probCCSResRule:] In this case $ S = \probCCSRes{P}{A} $, $ P \labelledstep{u} \Delta_P = \sum_{i \in I} p_i P_i $, $ u \notin A \cup \out{A} $, and $ \Delta_S = \probCCSRes{\Delta_P}{A} = \sum_{i \in I} p_i {\left( \probCCSRes{P_i}{A} \right)} $.
			By Definition~\ref{def:encCCSPi}, then $ \EncPCCSPPi{S} = \probPiRes{A}{\EncPCCSPPi{P}} $ and $ \OuterEncoding{\Delta_S} = \sum_{i \in I} p_i \Context{}{}{\probPiRes{A}{\EncPCCSPPi{P_i}}} $.
			By the induction hypothesis, then $ P \labelledstep{u} \Delta_P $ implies $ \OuterEncoding{P}\set{ \xLongrightarrow[p_i]{\hat{u}} T_{i, P} }_{i \in I} $ and $ \OuterEncoding{\Delta_P} \equiv \sum_{i \in I}p_i T_{i, P} $.
			$ \OuterEncoding{S} $ can emulate the step $ S \labelledstep{u} \Delta_S $ using the Rules \probPiResRule and \probPiParLRule to apply the steps in $ \OuterEncoding{P}\set{ \xLongrightarrow[p_i]{\hat{u}} T_{i, P} }_{i \in I} $ such that:
			\begin{align*}
				\OuterEncoding{S}\set{ \xLongrightarrow[p_i]{\hat{u}} T_i }_{i \in I} \quad \text{ where } T_i = \Context{}{}{\probPiRes{A}{T_{i, P}'}} \text{ and } T_{i, P} = \Context{}{}{T_{i, P}'}
			\end{align*}
			Because $ \OuterEncoding{\Delta_P} \equiv \sum_{i \in I}p_i T_{i, P} $ and $ \OuterEncoding{\Delta_P} = \sum_{i \in I}p_i \Context{}{}{\EncPCCSPPi{P_i}} $, then $ \OuterEncoding{\Delta_S} \equiv \sum_{i \in I} p_i T_i $.
		\item[\probCCSRelabelRule:] In this case $ S = P{\left[ f \right]} $, $ P \labelledstep{v} \Delta_P = \sum_{i \in I} p_i P_i $, $ f(v) = u $, and also $ \Delta_S = \Delta_P{\left[ f \right]} = \sum_{i \in I} p_i {\left( P_i{\left[ f \right]} \right)} $.
			Further, by Definition~\ref{def:encCCSPi}, it follows $ \EncPCCSPPi{S} = \EncPCCSPPi{P}\sub{\mathsf{ran}_f}{\mathsf{dom}_f} $ as well as $ \OuterEncoding{\Delta_S} = \sum_{i \in I} p_i \Context{}{}{\EncPCCSPPi{P_i}\sub{\mathsf{ran}_f}{\mathsf{dom}_f}} $.
			By induction hypothesis, then $ P \labelledstep{v} \Delta_P $ implies $ \OuterEncoding{P}\set{ \xLongrightarrow[p_i]{\hat{v}} T_{i, P} }_{i \in I} $ and $ \OuterEncoding{\Delta_P} \equiv \sum_{i \in I}p_i T_{i, P} $.
			$ \OuterEncoding{S} $ can emulate the step $ S \labelledstep{u} \Delta_S $ using the Rules \probPiResRule and \probPiParLRule to apply the steps in $ \OuterEncoding{P}\set{ \xLongrightarrow[p_i]{\hat{v}} T_{i, P} }_{i \in I} $ such that:
			\begin{align*}
				\OuterEncoding{S}\set{ \xLongrightarrow[p_i]{\hat{u}} T_i }_{i \in I} \quad \text{ where } T_i = \Context{}{}{T_{i, P}'\sub{\mathsf{ran}_f}{\mathsf{dom}_f}} \text{ and } T_{i, P} = \Context{}{}{T_{i, P}'}
			\end{align*}
			Because $ \OuterEncoding{\Delta_P} \equiv \sum_{i \in I}p_i T_{i, P} $ and $ \OuterEncoding{\Delta_P} = \sum_{i \in I}p_i \Context{}{}{\EncPCCSPPi{P_i}} $, then $ \OuterEncoding{\Delta_S} \equiv \sum_{i \in I} p_i T_i $.
		\item[\probCCSRecRule:] In this case $ S = C{\left\langle \tilde{y} \right\rangle} $, $ C \stackrel{\text{def}}{=} {\left( \tilde{x} \right)}P $, and $ \Delta_S = \pointDis{P\sub{\tilde{y}}{\tilde{x}}} $.
			By Definition~\ref{def:encCCSPi}, then we have $ \EncPCCSPPi{S} = \out{C}{\left( \tilde{y} \right)} $ and $ \OuterEncoding{\Delta_S} = \pointDis{\Context{}{}{\EncPCCSPPi{P\sub{\tilde{y}}{\tilde{x}}}}} $.
			$ \OuterEncoding{S} $ can emulate the step $ S \labelledstep{\tau} \Delta_S $ using the Rules \probPiResRule, \probPiParLRule, \probPiRepRule, \probPiOutRule, and \probPiComRule to reduce the replicated input in the outer encoding such that:
			\begin{align*}
				\OuterEncoding{S}\set{ \xlongrightarrow[1]{\tau} T } \quad \text{ where } T = \Context{}{}{\nul \mid \EncPCCSPPi{P}\sub{\RenamingPCCSPPi{\tilde{y}}}{\RenamingPCCSPPi{\tilde{x}}}}
			\end{align*}
			By Lemma~\ref{lem:nameInvariance}, then $ \OuterEncoding{\Delta_S} \equiv \pointDis{T} $.
	\end{description}
	For $ u = \tau $ we obtain from the above induction that $ S \step \Delta_S $ implies the existence of some $ \Delta_T = \sum_{i \in I} p_i T_i $ such that $ \OuterEncoding{S} \transstep \Delta_T $ and $ \OuterEncoding{\Delta_S} \equiv \Delta_T $, where $ \OuterEncoding{S} \transstep \Delta_T $ is a non-empty and finite sequence of steps.
	The proof of this lemma then is by induction on the number of steps in the source term sequence $ S \transstep \Delta_S $.
\end{proof}

For soundness, we have to prove that the encoding does not introduce any new behaviour.
Therefore, we show that every sequence of steps on the target belongs to a matching sequence of steps on the source.

\begin{lemma}[Weak \poc, Weak Soundness, $ \outerEncoding $/$ \encPCCSPPi $]
	\label{lem:soundness}
	\begin{align*}
		& \forall S, \Delta_T \logdot \OuterEncoding{S} \transstep \Delta_T \text{ implies }\\
		& \left( \exists \Delta_S', \Delta_T' \logdot S \transstep \Delta_S' \wedge \Delta_T \transstep \Delta_T' \wedge \OuterEncoding{\Delta_S'} \equiv \Delta_T' \right)
	\end{align*}
\end{lemma}

\begin{proof}
	We strengthen the proof goal for the induction, by assuming that the sequence $ \Delta_T \transstep \Delta_T' $ contains only \secondstep-steps.
	The proof is by induction on the number of steps in $ \OuterEncoding{S} \transstep \Delta_T $.
	The base case for zero steps, \ie $ \Delta_T = \pointDis{\OuterEncoding{S}} $, holds trivially by choosing $ \Delta_S = \pointDis{S} $ and $ \Delta_T' = \Delta_T $ such that $ \OuterEncoding{\Delta_S} = \Delta_T' $.
	For the induction step, assume $ \OuterEncoding{S} \transstep \Delta_T^* \step \Delta_T $.
	By the induction hypothesis, there are some $ \Delta_S^{**} $ and $ \Delta_T^{**} $ such that $ S \transstep \Delta_S^{**} $, $ \Delta_T^* \transstep \Delta_T^{**} $, and $ \OuterEncoding{\Delta_S^{**}} \equiv \Delta_T^{**} $, where the sequence $ \Delta_T^* \transstep \Delta_T^{**} $ contains only \secondstep-steps.
	Note that, by Definition~\ref{def:encCCSPi} and because of the renaming policy $ \renamingPCCSPPi $, the restriction of $ \encNameI $ ensures that no other step on the target can be in conflict with an \secondstep-step.
	Because of that, we can combine the steps in $ \Delta_T^* \transstep \Delta_T^{**} $ and $ \Delta_T^* \step \Delta_T $ to the sequence $ \Delta_T^* \step \Delta_T \transstep \Delta_T^{***} $, where $ \Delta_T \transstep \Delta_T^{***} $ is the result of removing the step $ \Delta_T^* \step \Delta_T $ from $ \Delta_T^* \transstep \Delta_T^{**} $ if it is contained in this sequence and then reordering the steps such that the remaining steps in $ \Delta_T^* \transstep \Delta_T^{**} $ are applied after the step $ \Delta_T^* \step \Delta_T $.
	We have to proof that there are some $ \Delta_S' $ and $ \Delta_T' $ such that $ S \transstep \Delta_S' $, $ \Delta_T \transstep \Delta_T' $, and $ \OuterEncoding{\Delta_S'} \equiv \Delta_T' $, where the sequence $ \Delta_T \transstep \Delta_T' $ contains only \secondstep-steps.
	Therefore, we construct $ \Delta_S^{**} \transstep \Delta_S' $ and the sequence $ \Delta_T^{***} \transstep \Delta_T' $ containing only \secondstep-steps such that $ S \transstep \Delta_S^{**} \transstep \Delta_S' $, $ \Delta_T^* \step \Delta_T \transstep \Delta_T^{***} \transstep \Delta_T' $, and $ \OuterEncoding{\Delta_S'} \equiv \Delta_T' $.

	By Definition~\ref{def:stepDistributions}, $ \Delta_T^* = \sum_{i \in I} p_i T_i^* $, $ \Delta_T = \sum_{i \in I} p_i \cdot \Delta_{T, i} $, $ \sum_{i \in I} p_i = 1 $, and $ T_i^* \step \Delta_{T, i} $ or $ \Delta_{T, i} = \pointDis{T_i^*} $ for all $ i \in I $.
	We perform a case split on the nature of $ T_i^* \step \Delta_{T, i} $ for all $ i \in I $ with $ \Delta_{T, i} \neq \pointDis{T_i^*} $ to generate the initially empty sets $ \mathcal{S} $ and $ \mathcal{T} $ of source and target term steps.
	We use $ \mathcal{S} $ and $ \mathcal{T} $ to collect the steps that we need for the sequences $ \Delta_S^{**} \transstep \Delta_S' $ and $ \Delta_T^{***} \transstep \Delta_T' $.
	\begin{description}
		\item[$ \Delta_T^* \step \Delta_ T $ is an \firststep-step:] By Definition~\ref{def:firststep}, then $ \Delta_T^* \step \Delta_ T $ is a communication step on a translated source term name $ x $.
			To complete the emulation of the corresponding source term communication on $ x $, we need to perform the \secondstep-step that was enabled by this \firststep-step.
			Accordingly, we add the respective source term step on $ x $ to $ \mathcal{S} $ and the \secondstep-step that was enabled by $ \Delta_T^* \step \Delta_ T $ to $ \mathcal{T} $.
		\item[$ \Delta_T^* \step \Delta_ T $ is an \secondstep-step:] By Definition~\ref{def:secondstep}, then $ T_i^* \step \Delta_{T, i} $ is a communication step on an instance of the reserved name $ \encNameI $.
			By Definition~\ref{def:encCCSPi} and Definition~\ref{def:firststep}, all in- and outputs on $ \encNameI $ are initially guarded in the encoding and can only be unguarded by an \firststep-step.
			Accordingly, $ \OuterEncoding{S} \transstep \Delta_T^* $ contains the corresponding \firststep-step that unguarded the input on $ \encNameI $ reduced in $ T_i^* \step \Delta_{T, i} $.
			Since $ \OuterEncoding{\Delta_S^{**}} \equiv \Delta_T^{**} $, then $ S \transstep \Delta_S^{**} $ already contains the corresponding communication step in the source, \ie in this case we do not have to add any steps to $ \mathcal{S} $ or $ \mathcal{T} $.
		\item[$ \Delta_T^* \step \Delta_ T $ is a \taustep-step:] By Definition~\ref{def:taustep}, then $ T_i^* \step \Delta_{T, i} $ is a communication step on an instance of the reserved name $ \encNameTau $.
			In this case we add the source term $ \tau $-step that is emulated by $ T_i^* \step \Delta_{T, i} $ to $ \mathcal{S} $ and leave $ \mathcal{T} $ unchanged.
		\item[$ \Delta_T^* \step \Delta_ T $ is a \repstep-step:] By Definition~\ref{def:repstep}, then $ \Delta_T^* \step \Delta_ T $ reduce a replicated input to emulate the unfolding of recursion.
			Again we add the source term step to unfold a recursion that is emulated by $ T_i^* \step \Delta_{T, i} $ to $ \mathcal{S} $ and leave $ \mathcal{T} $ unchanged.
		\item[Otherwise:] By Definition~\ref{def:encCCSPi}, all steps of an encoded source term are \firststep-steps, \secondstep-steps, \taustep-steps, or \repstep-steps.
	\end{description}
	We observe that in all cases, we need at most one step in the source and target.
	Since $ \Delta_T^* \transstep \Delta_T^{**} $ contains only \secondstep-steps, it can only complete the emulation of source terms steps and not start new such emulations.
	Because of $ \OuterEncoding{\Delta_S^{**}} \equiv \Delta_T^{**} $ and because the emulations of the source terms steps in $ \mathcal{S} $ were enabled in $ \Delta_T^{*} $, the source term steps in $ \mathcal{S} $ are enabled in $ \Delta_S^{**} $.
	If $ \mathcal{S} = \emptyset $ then we can choose $ \Delta_S' = \Delta_S^{**} $.
	Else $ \Delta_S^{**} \step \Delta_S' $ be the result of applying for each branch $ i \in I $ in the distribution $ \Delta_S^{**} $ with a step in $ \mathcal{S} $ the corresponding step.
	Similarly, if $ \mathcal{T} = \emptyset $ then $ \Delta_T' = \Delta_T^{***} $ and else let $ \Delta_T^{***} \step \Delta_T' $ apply the steps in $ \mathcal{T} $ on the respective branches in $ \Delta_T^{***} $.
	By Definition~\ref{def:encCCSPi}, $ \outerEncoding $ may produce $ \nul $ or $ \probPiRes{x}{\nul} $ as junk, \ie as leftovers from a completed emulation.
	However, all forms of junk produced by $ \outerEncoding $ are $ \nul $ and superfluous restrictions and are not observable modulo $ \equiv $.
	Since all initiated emulation attempts are completed and since $ \Delta_S' $ results from performing all source term steps emulated in the target, $ \OuterEncoding{\Delta_S'} \equiv \Delta_T' $.
\end{proof}

Divergence reflection ensures that the encoding cannot introduce new sources of divergence.

\begin{lemma}[Divergence Reflection, $ \outerEncoding $/$ \encPCCSPPi $]
	\label{lem:divergenceReflection}
	$ $\\
	For every $ S $, $ \OuterEncoding{S} \infiniteSteps $ implies $ S \infiniteSteps $.
\end{lemma}

\begin{proof}
	By Lemma~\ref{lem:soundness}, for every sequence $ \OuterEncoding{S} \transstep \Delta_T $ there are some $ \Delta_S' $ and $ \Delta_T' $ such that $ S \transstep \Delta_S' $, $ \Delta_T \transstep \Delta_T' $, and $ \EncPCCSPPi{\Delta_S} \equiv \Delta_T $, where in the proof of Lemma~\ref{lem:soundness} we additionally show that the sequence $ \Delta_T \transstep \Delta_T' $ is a sequence of \secondstep-steps.
	Moreover, from the construction of $ S \transstep \Delta_S' $, this sequence contains exactly one source term step for every \firststep-step, every \taustep-step, and every \repstep-step in $ \OuterEncoding{S} \transstep \Delta_T $.
	The \secondstep-steps are initially guarded, can be unguarded only by an \firststep-step, and for each \firststep-step exactly one \secondstep-step is unguarded.
	We conclude that \firststep-steps, \secondstep-steps, \taustep-steps, and \repstep-steps cannot introduce new loops.
	Since there are no other kinds of steps, this ensures divergence reflection.
\end{proof}

Success sensitiveness ensures that the translation passes a test if and only if the source term passes this test.

\begin{lemma}[Success Sensitiveness, $ \outerEncoding $/$ \encPCCSPPi $]
	\label{lem:successSensitiveness}
	$ $\\
	For every $ S $, $ \ReachBarb{S}{\success} $ iff $ \ReachBarb{\OuterEncoding{S}}{\success} $.
\end{lemma}

\begin{proof}
	By Definition~\ref{def:encCCSPi}, $ \HasBarb{S^*}{\success} $ iff $\HasBarb{\OuterEncoding{S^*}}{\success} $ for all $ S^* $.
	Then also $ \HasBarb{\Delta_S^*}{\success} $ iff $\HasBarb{\OuterEncoding{\Delta_S^*}}{\success} $ for all distributions $ \Delta_S^* $.
	\begin{itemize}
		\item If $ \ReachBarb{S}{\success} $, then $ S \transstep \Delta_S $ and $ \HasBarb{\Delta_S}{\success} $.
			By Lemma~\ref{lem:completeness}, then $ \OuterEncoding{S} \transstep \Delta_T $ and $ \OuterEncoding{\Delta_S} \equiv \Delta_T $.
			By Definition~\ref{def:encCCSPi}, $ \HasBarb{\Delta_S}{\success} $ implies $ \HasBarb{\OuterEncoding{\Delta_S}}{\success} $.
			By Lemma~\ref{lem:successSensitivenessTargetRelation}, then $ \OuterEncoding{\Delta_S} \equiv \Delta_T $ implies $ \HasBarb{\Delta_T}{\success} $.
			Finally, $ \OuterEncoding{S} \transstep \Delta_T $ and $ \HasBarb{\Delta_T}{\success} $ imply $ \ReachBarb{\OuterEncoding{S}}{\success} $.
		\item If $ \ReachBarb{\OuterEncoding{S}}{\success} $, then $ \OuterEncoding{S} \transstep \Delta_T $ and $ \HasBarb{\Delta_T}{\success} $.
			By Lemma~\ref{lem:soundness}, then $ S \transstep \Delta_S' $, $ \Delta_T \transstep \Delta_T' $, and $ \OuterEncoding{\Delta_S'} \equiv \Delta_T' $.
			Because of $ \Delta_T \transstep \Delta_T' $, $ \HasBarb{\Delta_T}{\success} $ implies $ \HasBarb{\Delta_T'}{\success} $.
			By Lemma~\ref{lem:successSensitivenessTargetRelation}, $ \HasBarb{\Delta_T'}{\success} $ and $ \OuterEncoding{\Delta_S'} \equiv \Delta_T' $ imply $ \HasBarb{\OuterEncoding{\Delta_S'}}{\success} $.
			By Definition~\ref{def:encCCSPi}, then $ \HasBarb{\Delta_S'}{\success} $.
			Finally, $ S \transstep \Delta_S' $ and $ \HasBarb{\Delta_S'}{\success} $ imply $ \ReachBarb{S}{\success} $. \qed
	\end{itemize}
\end{proof}

\begin{theorem}
	\label{thm:propertiesEnc}
	The encoding $ \outerEncoding $ satisfies weak compositionality, name invariance, weak probabilistic operational correspondence \wrt $ \equiv $, divergence reflection, and success sensitiveness.
\end{theorem}

\begin{proof}[Proof of Theorem~\ref{thm:propertiesEnc}]
	The proof is by the Lemmata~\ref{lem:weakCompositionality}, \ref{lem:nameInvariance}, \ref{lem:completeness}, \ref{lem:soundness}, \ref{lem:divergenceReflection}, and \ref{lem:successSensitiveness}, where Lemma~\ref{lem:successSensitivenessTargetRelation} proves that $ \equiv $ is success sensitive.
\end{proof}


\section{Weak Probabilistic Operational Correspondence}
\label{app:weakPOC}

\begin{definition}[Correspondence Simulation, \cite{PetersGlabbeek15}]
	\label{def:correspondenceSimulation}
	A relation $ \relation $ is a \emph{(weak reduction) correspondence simulation} if for each $ (P, Q) \in \relation $:
	\begin{itemize}
		\item $ P \transstep P' $ implies $ \exists Q' \logdot Q \transstep Q' \wedge (P', Q') \in \relation $
		\item $ Q \transstep Q' $ implies $ \exists P'', Q'' \logdot P \transstep P'' \wedge Q' \transstep Q'' \wedge (P'', Q'') \in \relation $
	\end{itemize}
	Two terms are \emph{correspondence similar} if a correspondence simulation relates them.
\end{definition}

A probabilistic version of correspondence simulation for a relation between probability distributions can be derived straightforwardly from Definition~\ref{def:correspondenceSimulation}.

\begin{definition}[Probabilistic Correspondence Simulation]
	\label{def:probabilisticCorrespondenceSimulation}
	A relation $ \relation $ is a \emph{(weak) probabilistic (reduction) correspondence simulation} if for each $ (P, Q) \in \relation $:
	\begin{itemize}
		\item $ P \transstep \Delta $ implies $ \exists \Theta \logdot Q \transstep \Theta \wedge (\Delta, \Theta) \in \pointDis{\relation} $
		\item $ Q \transstep \Theta $ implies $ \exists \Delta', \Theta' \logdot P \transstep \Delta' \wedge \Theta \transstep \Theta' \wedge (\Delta', \Theta') \in \pointDis{\relation} $
	\end{itemize}
	Two terms are \emph{probabilistic correspondence similar} if a probabilistic correspondence simulation relates them.
\end{definition}

\begin{definition}[Probabilistic Correspondence Simulation on Distributions]
	\label{def:probabilisticCorrespondenceSimulationDistributions}
	A relation $ \relation $ on distributions is a \emph{(weak) probabilistic (reduction) correspondence simulation} if for each $ (\Delta, \Theta) \in \relation $:
	\begin{itemize}
		\item $ \Delta \transstep \Delta' $ implies $ \exists \Theta' \logdot \Theta \transstep \Theta' \wedge (\Delta', \Theta') \in \relation $
		\item $ \Theta \transstep \Theta' $ implies $ \exists \Delta'', \Theta'' \logdot \Delta \transstep \Delta'' \wedge \Theta' \transstep \Theta'' \wedge (\Delta'', \Theta'') \in \relation $
	\end{itemize}
	Two terms are \emph{probabilistic correspondence similar} if a probabilistic correspondence simulation relates them.
\end{definition}

If $ \relation $ is a preorder and a probabilistic correspondence simulation then so is $ \pointDis{\relation} $.

\begin{lemma}[Preservation of the Correspondence Property]
	\label{lem:preserveProbabilisticCorrespondenceSimulation}
	$ $\\
	If the preorder $ \relation $ is a probabilistic correspondence simulation then so is $ \pointDis{\relation} $.
\end{lemma}

\begin{proof}[Proof of Lemma~\ref{lem:preserveProbabilisticCorrespondenceSimulation}]
	Assume a probabilistic correspondence simulation $ \relation $ that is a preorder and $ \left( \Delta, \Theta \right) \in \pointDis{\relation} $.
	Since $ \relation $ is a preorder and by the Lemmata~\ref{lem:reflexivityRelationDistiribution} and \ref{lem:transitvityRelationDistiribution}, $ \pointDis{\relation} $ is a preorder.
	By Definition~\ref{def:relationsDistributions}, then there is some index set $ I $ such that $ \Delta = \sum_{i \in I} p_i \pointDis{P_i} $, $ \Theta = \sum_{i \in I} p_i \pointDis{Q_i} $, $ \sum_{i \in I} p_i = 1 $, and for each $ i \in I $ we have $ \left( P_i, Q_i \right) \in \relation $.
	\begin{itemize}
		\item Assume $ \Delta \transstep \Delta' $.
			If $ \Delta' = \Delta $ then we can choose $ \Theta' = \Theta $ such that $ \Theta \transstep \Theta' $ and $ \left( \Delta', \Theta' \right) \in \pointDis{\relation} $.
			Else, by Definition~\ref{def:stepDistributions}, then $ \Delta' = \sum_{i \in I} p_i \Delta_i' $ and for some (at least one) $ i \in I $ a processes $ P_i $ in $ \Delta $ performed a sequence of at least one step $ P_i \transstep \Delta_i' $.
			For all other $ i $ we have $ \Delta_i' = \pointDis{P_i} $.
			Since $ \relation $ is a probabilistic correspondence simulation, $ \left( P_i, Q_i \right) \in \relation $ and $ P_i \step \Delta_i' $ imply $ Q_i \transstep \Theta_i' $ and $ \left( \Delta_i', \Theta_i' \right) \in \pointDis{\relation} $.
			For the $ i \in I $ without a step, we choose $ \Theta_i' = \pointDis{Q_i} $ such that $ Q_i \transstep \Theta_i' $ and $ \left( \Delta_i', \Theta_i' \right) \in \pointDis{\relation} $.
			Then $ \Theta \transstep \Theta' = \sum_{i \in I} p_i \Theta_i' $ and $ \left( \Delta_i', \Theta_i' \right) \in \pointDis{\relation} $.
		\item Assume $ \Theta \transstep \Theta' $.
			If $ \Theta' = \Theta $ then we can choose $ \Delta'' = \Delta $ and $ \Theta'' = \Theta $ such that $ \Delta \transstep \Delta'' $, $ \Theta' \transstep \Theta'' $, and $ \left( \Delta'', \Theta'' \right) \in \pointDis{\relation} $.
			Else, by Definition~\ref{def:stepDistributions}, then $ \Theta' = \sum_{i \in I} p_i \Theta_i' $ and for some (at least one) $ i \in I $ a processes $ Q_i $ in $ \Theta $ performed a sequence of at least one step $ Q_i \transstep \Theta_i' $.
			For all other $ i $ we have $ \Theta_i' = \pointDis{Q_i} $.
			Since $ \relation $ is a probabilistic correspondence simulation, $ \left( P_i, Q_i \right) \in \relation $ and $ Q_i \step \Theta_i' $ imply $ P_i \transstep \Delta_i'' $, $ \Theta_i' \transstep \Theta_i'' $, and $ \left( \Delta_i'', \Theta_i'' \right) \in \pointDis{\relation} $.
			For the $ i \in I $ without a step, we choose $ \Delta_i'' = \pointDis{P_i} $ and $ \Theta_i'' = \Theta_i $ such that $ P_i \transstep \Delta_i'' $, $ \Theta_i' \transstep \Theta_i'' $ and $ \left( \Delta_i'', \Theta_i'' \right) \in \pointDis{\relation} $.
			Then $ \Delta \transstep \Delta'' = \sum_{i \in I} p_i \Delta_i'' $, $ \Theta' \transstep \Theta'' = \sum_{i \in I} p_i \Theta_i'' $ and $ \left( \Delta_i'', \Theta_i'' \right) \in \pointDis{\relation} $.
	\end{itemize}
	We conclude that $ \pointDis{\relation} $ is a probabilistic correspondence simulation.
\end{proof}

\begin{theorem}[Weak \poc]
	\label{thm:weakPOC}
	$ \enc{\cdot} $ is weakly probabilistically operationally corresponding \wrt a preorder $ \relationT \subseteq \processes{\target}^2 $ that is a probabilistic correspondence simulation iff \\ $ \exists \relationEnc \logdot \left( \forall S \logdot \left( S, \enc{S} \right) \in \relationEnc \right) \wedge \relationT = \restrictRelation{\relationEnc}{\processes{\target}} \wedge \left( \forall S, T \logdot \left( S, T \right) \in \relationEnc \longrightarrow \left( \enc{S}, T \right) \in \relationT \right) \wedge \relationEnc $ is a preorder and a probabilistic correspondence simulation.
\end{theorem}

One of the condition in Theorem~\ref{thm:weakPOC} is that the relation $ \relationT $ on the target is obtained from the induced relation $ \relationEnc $ between source and target by reduction on target terms, \ie $ \relationT = \restrictRelation{\relationEnc}{\processes{\target}} $.
We prove that this property is preserved by the lift operation in Definition~\ref{def:relationsDistributions}.

\begin{lemma}
	\label{lem:preservationReduction}
	If $ \relationT = \restrictRelation{\relationEnc}{\processes{\target}} $ then $ \pointDis{\relationT} = \restrictRelation{\pointDis{\relationEnc}}{\processes{\target}} $.
\end{lemma}

\begin{proof}
	Assume $ \relationT = \restrictRelation{\relationEnc}{\processes{\target}} $.
	By Definition~\ref{def:relationsDistributions}, $ \left( \Delta, \Theta \right) \in \pointDis{\relationEnc} $ if
	\begin{enumerate}[(i)]
		\item $ \Delta = \sum_{i \in I} p_i \pointDis{P_i} $, where $ I $ is a finite index set and $ \sum_{i \in I} p_i = 1 $,
		\item for each $ i \in I $ there is a process $ Q_i $ such that $ \left( P_i, Q_i \right) \in \relationEnc $, and
		\item $ \Theta = \sum_{i \in I} p_i \pointDis{Q_i} $.
	\end{enumerate}
	Then $ \left( \Delta, \Theta \right) \in \restrictRelation{\pointDis{\relationEnc}}{\processes{\target}} $ if
	\begin{enumerate}[(i)]
		\item $ \Delta = \sum_{i \in I} p_i \pointDis{P_i} $ in the target, where $ I $ is a finite index set and $ \sum_{i \in I} p_i = 1 $,
		\item for each $ i \in I $ there is a process $ Q_i $ such that $ \left( P_i, Q_i \right) \in \restrictRelation{\relationEnc}{\processes{\target}} $, and
		\item $ \Theta = \sum_{i \in I} p_i \pointDis{Q_i} $ in the target.
	\end{enumerate}
	Since $ \relationT = \restrictRelation{\relationEnc}{\processes{\target}} $ and by Definition~\ref{def:relationsDistributions}, then $ \pointDis{\relationT} = \restrictRelation{\pointDis{\relationEnc}}{\processes{\target}} $.
\end{proof}

The condition $ \relationT = \restrictRelation{\relationEnc}{\processes{\target}} $ in Theorem~\ref{thm:weakPOC} allows us to prove that $ \relationT $ is a preorder if $ \relationEnc $ is a preorder.

\begin{lemma}
	\label{lem:relationTIsPreorder}
	If $ \relationEnc $ is a preorder and $ \relationT = \restrictRelation{\relationEnc}{\processes{\target}} $ then $ \relationT $ is a preorder.
\end{lemma}

\begin{proof}
	Assume a preorder $ \relationEnc $ and $ \relationT = \restrictRelation{\relationEnc}{\processes{\target}} $.
	\begin{description}
		\item[Reflexivity:] Since $ \relationEnc $ is reflexive, $ \left( T, T \right) \in \relationEnc $ for all target terms $ T $.
			Because of $ \relationT = \restrictRelation{\relationEnc}{\processes{\target}} $, then $ \left( T, T \right) \in \relationT $.
		\item[Transitivity:] Assume $ \left( T_1, T_2 \right) \in \relationT $ and $ \left( T_2, T_3 \right) \in \relationT $ for some target terms $ T_1 $, $ T_2 $, and $ T_3 $.
			Because of $ \relationT = \restrictRelation{\relationEnc}{\processes{\target}} $, then $ \left( T_1, T_2 \right) \in \relationEnc $ and $ \left( T_2, T_3 \right) \in \relationEnc $.
			Since $ \relationEnc $ is transitive, then $ \left( T_1, T_3 \right) \in \relationEnc $.
			Because of $ \relationT = \restrictRelation{\relationEnc}{\processes{\target}} $, then $ \left( T_1, T_3 \right) \in \relationT $.
	\end{description}
	We conclude that $ \relationT $ is a preorder.
\end{proof}

The condition $ \forall S, T \logdot \left( S, T \right) \in \relationEnc \longrightarrow \left( \enc{S}, T \right) \in \relationT $ is necessary to ensure (with the remaining properties) that the encoding satisfies weak \poc in the 'only if'-case of Theorem~\ref{thm:weakPOC}.
Therefore we lift this property to distributions.

\begin{lemma}
	\label{lem:liftCondition}
	If $ \forall S, T \logdot \left( S, T \right) \in \relationEnc \longrightarrow \left( \enc{S}, T \right) \in \relationT $ then for all distributions $ \Delta_S $ on the source and all distributions $ \Delta_T $ on the target $ \left( \Delta_S, \Delta_T \right) \in \pointDis{\relationEnc} $ implies $ \left( \enc{\Delta_S}, \Delta_T \right) \in \pointDis{\relationT} $.
\end{lemma}

\begin{proof}
	Assume $ \forall S, T \logdot \left( S, T \right) \in \relationEnc \longrightarrow \left( \enc{S}, T \right) \in \relationT $ and $ \left( \Delta_S, \Delta_T \right) \in \pointDis{\relationEnc} $.
	By Definition~\ref{def:relationsDistributions}, then there is some index set $ I $ such that $ \Delta_S = \sum_{i \in I} p_i \pointDis{S_i} $, $ \Delta_T = \sum_{i \in I} p_i \pointDis{T_i} $, $ \sum_{i \in I} p_i = 1 $, and for each $ i \in I $ we have $ \left( S_i, T_i \right) \in \relationEnc $.
	With $ \forall S, T \logdot \left( S, T \right) \in \relationEnc \longrightarrow \left( \enc{S}, T \right) \in \relationT $, then $ \left( \enc{S_i}, T_i \right) \in \relationT $ for each $ i \in I $.
	By Definition~\ref{def:relationsDistributions}, then $ \left( \enc{\Delta_S}, \Delta_T \right) \in \pointDis{\relationT} $.
\end{proof}

To prove Theorem~\ref{thm:weakPOC}, it is necessary to ensure that $ \relationT $ is a probabilistic correspondence simulation if $ \relationEnc $ is.

\begin{lemma}
	\label{lem:relationTIsCorrespondence}
	If $ \relationEnc $ is a probabilistic correspondence simulation and $ \relationT = \restrictRelation{\relationEnc}{\processes{\target}} $ then $ \relationT $ is a probabilistic correspondence simulation.
\end{lemma}

\begin{proof}
	Assume that $ \relationEnc $ is a probabilistic correspondence simulation and $ \relationT = \restrictRelation{\relationEnc}{\processes{\target}} $.
	Moreover, assume $ \left( T_1, T_2 \right) \in \relationT $.
	Because of $ \relationT = \restrictRelation{\relationEnc}{\processes{\target}} $, then $ \left( T_1, T_2 \right) \in \relationEnc $.
	\begin{description}
		\item[Case (i):] Assume $ T_1 \transstep \Delta $.
			Since $ \relationEnc $ is a probabilistic correspondence simulation, then $ T_2 \transstep \Theta $ and $ \left( \Delta, \Theta \right) \in \pointDis{\relationEnc} $.
			With $ \relationT = \restrictRelation{\relationEnc}{\processes{\target}} $ and Lemma~\ref{lem:preservationReduction}, then $ \left( \Delta, \Theta \right) \in \pointDis{\relationT} $.
		\item[Case (ii):] Assume $ T_2 \transstep \Theta $.
			Since $ \relationEnc $ is a probabilistic correspondence simulation, then $ T_1 \transstep \Delta' $, $ \Theta \transstep \Theta' $, and $ \left( \Delta', \Theta' \right) \in \pointDis{\relationEnc} $.
			With $ \relationT = \restrictRelation{\relationEnc}{\processes{\target}} $ and Lemma~\ref{lem:preservationReduction}, then $ \left( \Delta', \Theta' \right) \in \pointDis{\relationT} $.
	\end{description}
	We conclude that $ \relationT $ is a probabilistic correspondence simulation.
\end{proof}

Next we show the main theorem of Section~\ref{app:weakPOC}.
Weak probabilistic operational correspondence induces a relation between source and target terms that is a probabilistic correspondence simulation.
More precisely, Theorem~\ref{thm:weakPOC} states:
\begin{quote}
	$ \enc{\cdot} $ is weakly probabilistically operationally corresponding \wrt a preorder $ \relationT \subseteq \processes{\target}^2 $ that is a probabilistic correspondence simulation iff $ \exists \relationEnc \logdot \left( \forall S \logdot \left( S, \enc{S} \right) \in \relationEnc \right) \wedge \relationT = \restrictRelation{\relationEnc}{\processes{\target}} $\\ $ \wedge \; \left( \forall S, T \logdot \left( S, T \right) \in \relationEnc \longrightarrow \left( \enc{S}, T \right) \in \relationT \right) \wedge \relationEnc $ is a preorder and a probabilistic correspondence simulation.
\end{quote}

\begin{proof}[Proof of Theorem~\ref{thm:weakPOC}]
	We prove the two directions of the result separately.
	\begin{description}
		\item[if ($ \longrightarrow $):] Assume that $ \enc{\cdot} $ is weakly probabilistically operationally corresponding \wrt a preorder $ \relationT \subseteq \processes{\target}^2 $ that is a probabilistic correspondence simulation.
			We construct $ \relationEnc $ from $ \relationT $ by adding $ \left( S, \enc{S} \right) $ for all source terms $ S $ and then building the reflexive and transitive closure.
			Accordingly, $ \forall S \logdot \left( S, \enc{S} \right) \in \relationEnc $ holds by construction.
			Since we did not add any pairs of only source terms, \ie no pairs of the form $ \left( S_1, S_2 \right) $ where both $ S_1 $ and $ S_2 $ are source terms, and since the only such pairs added by the reflexive and transitive closure are of the form $ \left( S, S \right) $, we have $ \relationT = \restrictRelation{\relationEnc}{\processes{\target}} $.
			Next we prove that $ \forall S, T \logdot \left( S, T \right) \in \relationEnc \longrightarrow \left( \enc{S}, T \right) \in \relationT $.
			Therefore, fix some $ S $ and $ T $ and assume $ \left( S, T \right) \in \relationEnc $.
			By the construction of $ \relationEnc $, $ \left( S, \enc{S} \right) \in \relationEnc $ and all pairs relating a source and a target term contain a source term and its literal translation or result from such a pair, $ \relationT $, and the transitive closure in the construction of $ \relationEnc $.
			Hence, $ T = \enc{S} $ or $ \left( \enc{S}, T \right) \in \relationT $.
			In the former case, $ \left( \enc{S}, T \right) \in \relationT $ follows from the reflexivity of $ \relationT $.
			The latter case directly provides $ \left( \enc{S}, T \right) \in \relationT $.
			That $ \relationEnc $ is a preorder directly follows from the construction of $ \relationEnc $, because we used the reflexive and transitive closure.
			As last condition we have to show that $ \relationEnc $ is a probabilistic correspondence simulation.
			By Definition~\ref{def:probabilisticCorrespondenceSimulation}, for all $ \left( P, Q \right) \in \relationEnc$ and all $ P \transstep \Delta $ we have to find $ Q \transstep \Theta $ such that $ \left( \Delta, \Theta \right) \in \relationEnc $ and for all $ Q \transstep \Theta $ we have to find $ P \transstep \Delta' $ and $ \Theta \transstep \Theta' $ such that $ \left( \Delta', \Theta' \right) \in \pointDis{\relationEnc} $.
			By the construction of $ \relationEnc $, $ P $ may be a source or target term, but $ Q $ is a target term.
			\begin{itemize}
				\item Assume $ P \transstep \Delta $.
					If $ P $ is a source term, then we have $ \enc{P} \transstep \Delta' $ with $ \left( \Delta, \Delta' \right) \in \pointDis{\relationT} $, because of completeness in weak \poc in Definition~\ref{def:weakPOC}.
					Since $ Q $ is a target term and by the construction of $ \relationEnc $, $ \left( P, Q \right) \in \relationEnc $ implies $ \left( \enc{P}, Q \right) \in \relationEnc $.
					Because of $ \relationT = \restrictRelation{\relationEnc}{\processes{\target}} $, then $ \left( \enc{P}, Q \right) \in \relationT $.
					Then $ Q \transstep \Theta $ and $ \left( \Delta', \Theta \right) \in \pointDis{\relationT} $, because $ \relationT $ is a probabilistic correspondence simulation.
					With the transitivity of $ \relationT $ and thus $ \pointDis{\relationT} $, $ \left( \Delta, \Delta' \right) \in \pointDis{\relationT} $ and $ \left( \Delta', \Theta \right) \in \pointDis{\relationT} $ imply $ \left( \Delta, \Theta \right) \in \pointDis{\relationT} $.
					Finally, $ \left( \Delta, \Theta \right) \in \pointDis{\relationEnc} $ follows from $ \left( \Delta, \Theta \right) \in \pointDis{\relationT} $, by the construction of $ \relationEnc $.
					\\
					Else, assume that $ P $ is a target term.
					Then $ \left( P, Q \right) \in \relationEnc $ implies $ \left( P, Q \right) \in \relationT $, because $ Q $ is a target term and $ \relationT = \restrictRelation{\relationEnc}{\processes{\target}} $.
					Since $ \relationT $ is a probabilistic correspondence simulation (see Definition~\ref{def:probabilisticCorrespondenceSimulation}), then $ P \transstep \Delta $ implies $ Q \transstep \Theta $ with $ \left( \Delta, \Theta \right) \in \pointDis{\relationT} $.
					Finally, $ \left( \Delta, \Theta \right) \in \pointDis{\relationEnc} $ follows from $ \left( \Delta, \Theta \right) \in \pointDis{\relationT} $, by the construction of $ \relationEnc $.
				\item Assume $ Q \transstep \Theta $.
					If $ P $ is a source term, then $ \left( \enc{P}, Q \right) \in \relationEnc $, by the construction of $ \relationEnc $.
					Because of $ \relationT = \restrictRelation{\relationEnc}{\processes{\target}} $, then $ \left( \enc{P}, Q \right) \in \relationT $.
					Since $ \relationT $ is a probabilistic correspondence simulation (see second case of Definition~\ref{def:probabilisticCorrespondenceSimulation}), then $ Q \transstep \Theta $ implies $ \enc{P} \transstep \Delta_T' $, $ \Theta \transstep \Theta' $, and $ \left( \Delta_T', \Theta' \right) \in \pointDis{\relationT} $.
					From $ \enc{P} \transstep \Delta_T' $ we obtain $ P \transstep \Delta'' $ and $ \Delta_T' \transstep \Delta_T'' $ with $ \left( \Delta'', \Delta_T'' \right) \in \pointDis{\relationT} $, because of soundness in weak \poc in Definition~\ref{def:weakPOC}.
					By Lemma~\ref{lem:preserveProbabilisticCorrespondenceSimulation}, $ \pointDis{\relationT} $ is a probabilistic correspondence simulation.
					By the first case of Definition~\ref{def:probabilisticCorrespondenceSimulationDistributions}, then $ \left( \Delta_T', \Theta' \right) \in \pointDis{\relationT} $ and $ \Delta_T' \transstep \Delta_T'' $ imply $ \Theta' \transstep \Theta'' $ and $ \left( \Delta_T'', \Theta'' \right) \in \pointDis{\relationT} $.
					By transitivity, we obtain $ \left( \Delta'', \Theta'' \right) \in \pointDis{\relationT} $ and $ \Theta \transstep \Theta'' $.
					Finally, $ \left( \Delta'', \Theta'' \right) \in \pointDis{\relationEnc} $ follows from $ \left( \Delta'', \Theta'' \right) \in \pointDis{\relationT} $, by the construction of $ \relationEnc $.
					\\
					Else, assume that $ P $ is a target term.
					Then $ \left( P, Q \right) \in \relationEnc $ implies $ \left( P, Q \right) \in \relationT $, because $ Q $ is a target term and $ \relationT = \restrictRelation{\relationEnc}{\processes{\target}} $.
					Since $ \relationT $ is a probabilistic correspondence simulation (see Definition~\ref{def:probabilisticCorrespondenceSimulation}), then $ Q \transstep \Theta $ implies $ P \transstep \Delta' $, $ \Theta \transstep \Theta' $, and $ \left( \Delta', \Theta' \right) \in \pointDis{\relationT} $.
					Finally, $ \left( \Delta', \Theta' \right) \in \pointDis{\relationEnc} $ follows from $ \left( \Delta', \Theta' \right) \in \pointDis{\relationT} $, by the construction of $ \relationEnc $.
			\end{itemize}
		\item[only if ($ \longleftarrow $):] We assume that there is a relation $ \relationEnc $ such that $ \forall S \logdot \left( S, \enc{S} \right) \in \relationEnc $, $ \relationT = \restrictRelation{\relationEnc}{\processes{\target}} $, $ \forall S, T \logdot \left( S, T \right) \in \relationEnc \longrightarrow \left( \enc{S}, T \right) \in \relationT $, and $ \relationEnc $ is a preorder and a probabilistic correspondence simulation.
			We start with weak probabilistic operational correspondence.
			\begin{description}
				\item[Completeness:] Assume $ S \transstep \Delta_S $.
					Since $ \left( S, \enc{S} \right) \in \relationEnc $ and because $ \relationEnc $ is a probabilistic correspondence simulation, then $ \enc{S} \transstep \Delta_T $ and $ \left( \Delta_S, \Delta_T \right) \in \pointDis{\relationEnc} $.
					By $ \forall S, T \logdot \left( S, T \right) \in \relationEnc \longrightarrow \left( \enc{S}, T \right) \in \relationT $ and Lemma~\ref{lem:liftCondition}, then $ \left( \Delta_S, \Delta_T \right) \in \pointDis{\relationEnc} $ implies $ \left( \enc{\Delta_S}, \Delta_T \right) \in \pointDis{\relationT} $.
				\item[Weak Soundness:] Assume $ \enc{S} \transstep \Delta_T $.
					Since $ \left( S, \enc{S} \right) \in \relationEnc $ and because $ \relationEnc $ is a probabilistic correspondence simulation, then $ S \transstep \Delta_S $, $ \Delta_T \transstep \Delta_T' $, and $ \left( \Delta_S, \Delta_T' \right) \in \pointDis{\relationEnc} $.
					By $ \forall S, T \logdot \left( S, T \right) \in \relationEnc \longrightarrow \left( \enc{S}, T \right) \in \relationT $ and Lemma~\ref{lem:liftCondition}, $ \left( \Delta_S, \Delta_T' \right) \in \pointDis{\relationEnc} $ implies $ \left( \enc{\Delta_S}, \Delta_T' \right) \in \pointDis{\relationT} $.
			\end{description}
			By Definition~\ref{def:weakPOC}, then $ \enc{\cdot} $ is weakly probabilistically operationally corresponding \wrt $ \relationT $.
			Finally, since $ \relationEnc $ is a preorder and a probabilistic correspondence simulation and because of the Lemmata~\ref{lem:relationTIsPreorder} and \ref{lem:relationTIsCorrespondence}, $ \relationT $ is a preorder and a probabilistic correspondence simulation. \qed
	\end{description}
\end{proof}


\section{(Strong) Probabilistic Operational Correspondence}
\label{app:strongPOC}

\begin{definition}[Probabilistic Bisimulation]
	\label{def:probabilisticBisimulation}
	A relation $ \relation $ is a \emph{probabilistic (reduction) bisimulation} if for each $ (P, Q) \in \relation $:
	\begin{itemize}
		\item $ P \transstep \Delta $ implies $ \exists \Theta \logdot Q \transstep \Theta \wedge (\Delta, \Theta) \in \pointDis{\relation} $
		\item $ Q \transstep \Theta $ implies $ \exists \Delta \logdot P \transstep \Delta \wedge (\Delta, \Theta) \in \pointDis{\relation} $
	\end{itemize}
	Two terms are \emph{probabilistic bisimilar} if a probabilistic bisimulation relates them.
\end{definition}

We can reuse several of the auxiliary results derived in Section~\ref{app:weakPOC} for the proof of Theorem~\ref{thm:POC}.
But we have to adapt Lemma~\ref{lem:relationTIsCorrespondence} to probabilistic bisimulation.

\begin{lemma}
	\label{lem:relationTIsBisimulation}
	If $ \relationEnc $ is a probabilistic bisimulation and $ \relationT = \restrictRelation{\relationEnc}{\processes{\target}} $ then $ \relationT $ is a probabilistic bisimulation.
\end{lemma}

\begin{proof}
	Assume that $ \relationEnc $ is a probabilistic bisimulation and $ \relationT = \restrictRelation{\relationEnc}{\processes{\target}} $.
	Moreover, assume $ \left( T_1, T_2 \right) \in \relationT $.
	Because of $ \relationT = \restrictRelation{\relationEnc}{\processes{\target}} $, then $ \left( T_1, T_2 \right) \in \relationEnc $.
	\begin{description}
		\item[Case (i):] Assume $ T_1 \transstep \Delta $.
			Since $ \relationEnc $ is a probabilistic bisimulation, then $ T_2 \transstep \Theta $ and $ \left( \Delta, \Theta \right) \in \pointDis{\relationEnc} $.
			With $ \relationT = \restrictRelation{\relationEnc}{\processes{\target}} $ and Lemma~\ref{lem:preservationReduction}, then $ \left( \Delta, \Theta \right) \in \pointDis{\relationT} $.
		\item[Case (ii):] Assume $ T_2 \transstep \Theta $.
			Since $ \relationEnc $ is a probabilistic bisimulation, then $ T_1 \transstep \Delta $ and $ \left( \Delta, \Theta \right) \in \pointDis{\relationEnc} $.
			With $ \relationT = \restrictRelation{\relationEnc}{\processes{\target}} $ and Lemma~\ref{lem:preservationReduction}, then $ \left( \Delta, \Theta \right) \in \pointDis{\relationT} $.
	\end{description}
	We conclude that $ \relationT $ is a probabilistic bisimulation.
\end{proof}

\begin{definition}[Probabilistic Operational Correspondence]
	\label{def:POC}
	An encoding $\enc{\cdot} : \processes{\source} \to \processes{\target}$ is \emph{probabilistic operationally corresponding} (\poc) \wrt $\relationT \subseteq \processes{\target}^2 $ if it is:
	\begin{description}
		\item[\quad Probabilistic Complete:] $ $\\
		\hspace*{2em} $ \forall S, \Delta_S \logdot S \transstep \Delta_S \text{ implies } \left( \exists \Delta_T \logdot \enc{S} \transstep \Delta_T \wedge \left( \enc{\Delta_S}, \Delta_T \right) \in \pointDis{\relationT} \right) $
		\item[\quad Probabilistic Sound:] $ $\\
		\hspace*{2em} $ \forall S, \Delta_T \logdot \enc{S} \transstep \Delta_T \text{ implies } \left( \exists \Delta_S \logdot S \transstep \Delta_S \wedge \left( \enc{\Delta_S}, \Delta_T \right) \in \pointDis{\relationT} \right) $
	\end{description}
\end{definition}

\begin{theorem}[\poc]
	\label{thm:POC}
	$ \enc{\cdot} $ is probabilistically operationally corresponding \wrt a preorder $ \relationT \subseteq \processes{\target}^2 $ that is a probabilistic bisimulation iff\\
	$ \exists \relationEnc \logdot \left( \forall S \logdot \left( S, \enc{S} \right) \in \relationEnc \right) \wedge \relationT = \restrictRelation{\relationEnc}{\processes{\target}} \wedge \left( \forall S, T \logdot \left( S, T \right) \in \relationEnc \longrightarrow \left( \enc{S}, T \right) \in \relationT \right) \wedge \relationEnc $ is a preorder and a probabilistic bisimulation.
\end{theorem}

\begin{proof}[Proof of Theorem~\ref{thm:POC}]
	We prove the two directions of the result separately.
	\begin{description}
		\item[if ($ \longrightarrow $):] Assume that $ \enc{\cdot} $ is probabilistically operationally corresponding \wrt a preorder $ \relationT \subseteq \processes{\target}^2 $ that is a probabilistic bisimulation.
			We construct $ \relationEnc $ from $ \relationT $ by adding $ \left( S, \enc{S} \right) $ for all source terms $ S $ and then building the reflexive and transitive closure.
			Accordingly, $ \forall S \logdot \left( S, \enc{S} \right) \in \relationEnc $ holds by construction.
			Since we did not add any pairs of only source terms, \ie no pairs of the form $ \left( S_1, S_2 \right) $ where both $ S_1 $ and $ S_2 $ are source terms, and since the only such pairs added by the reflexive and transitive closure are of the form $ \left( S, S \right) $, we have $ \relationT = \restrictRelation{\relationEnc}{\processes{\target}} $.
			Next we prove that $ \forall S, T \logdot \left( S, T \right) \in \relationEnc \longrightarrow \left( \enc{S}, T \right) \in \relationT $.
			Therefore, fix some $ S $ and $ T $ and assume $ \left( S, T \right) \in \relationEnc $.
			By the construction of $ \relationEnc $, $ \left( S, \enc{S} \right) \in \relationEnc $ and all pairs relating a source and a target term contain a source term and its literal translation or result from such a pair, $ \relationT $, and the transitive closure in the construction of $ \relationEnc $.
			Hence, $ T = \enc{S} $ or $ \left( \enc{S}, T \right) \in \relationT $.
			In the former case, $ \left( \enc{S}, T \right) \in \relationT $ follows from the reflexivity of $ \relationT $.
			The latter case directly provides $ \left( \enc{S}, T \right) \in \relationT $.
			That $ \relationEnc $ is a preorder directly follows from the construction of $ \relationEnc $, because we used the reflexive and transitive closure.
			As last condition we have to show that $ \relationEnc $ is a probabilistic bisimulation.
			By Definition~\ref{def:probabilisticBisimulation}, for all $ \left( P, Q \right) \in \relationEnc $ and all $ P \transstep \Delta $ we have to find $ Q \transstep \Theta $ such that $ \left( \Delta, \Theta \right) \in \relationEnc $ and for all $ Q \transstep \Theta $ we have to find $ P \transstep \Delta $ such that $ \left( \Delta, \Theta \right) \in \pointDis{\relationEnc} $.
			By the construction of $ \relationEnc $, $ P $ may be a source or target term, but $ Q $ is a target term.
			\begin{itemize}
				\item Assume $ P \transstep \Delta $.
					If $ P $ is a source term, then we have $ \enc{P} \transstep \Delta' $ with $ \left( \Delta, \Delta' \right) \in \pointDis{\relationT} $, because of completeness in \poc in Definition~\ref{def:POC}.
					Since $ Q $ is a target term and by the construction of $ \relationEnc $, $ \left( P, Q \right) \in \relationEnc $ implies $ \left( \enc{P}, Q \right) \in \relationEnc $.
					Because of $ \relationT = \restrictRelation{\relationEnc}{\processes{\target}} $, then $ \left( \enc{P}, Q \right) \in \relationT $.
					Then $ Q \transstep \Theta $ and $ \left( \Delta', \Theta \right) \in \pointDis{\relationT} $, because $ \relationT $ is a probabilistic bisimulation.
					With the transitivity of $ \relationT $ and thus $ \pointDis{\relationT} $, $ \left( \Delta, \Delta' \right) \in \pointDis{\relationT} $ and $ \left( \Delta', \Theta \right) \in \pointDis{\relationT} $ imply $ \left( \Delta, \Theta \right) \in \pointDis{\relationT} $.
					Finally, $ \left( \Delta, \Theta \right) \in \pointDis{\relationEnc} $ follows from $ \left( \Delta, \Theta \right) \in \pointDis{\relationT} $, by the construction of $ \relationEnc $.
					\\
					Else, assume that $ P $ is a target term.
					Then $ \left( P, Q \right) \in \relationEnc $ implies $ \left( P, Q \right) \in \relationT $, because $ Q $ is a target term and $ \relationT = \restrictRelation{\relationEnc}{\processes{\target}} $.
					Since $ \relationT $ is a probabilistic bisimulation (see Definition~\ref{def:probabilisticBisimulation}), then $ P \transstep \Delta $ implies $ Q \transstep \Theta $ with $ \left( \Delta, \Theta \right) \in \pointDis{\relationT} $.
					Finally, $ \left( \Delta, \Theta \right) \in \pointDis{\relationEnc} $ follows from $ \left( \Delta, \Theta \right) \in \pointDis{\relationT} $, by the construction of $ \relationEnc $.
				\item Assume $ Q \transstep \Theta $.
					If $ P $ is a source term, then $ \left( \enc{P}, Q \right) \in \relationEnc $, by the construction of $ \relationEnc $.
					Because of $ \relationT = \restrictRelation{\relationEnc}{\processes{\target}} $, then $ \left( \enc{P}, Q \right) \in \relationT $.
					Since $ \relationT $ is a probabilistic bisimulation (see second case of Definition~\ref{def:probabilisticBisimulation}), then $ Q \transstep \Theta $ implies $ \enc{P} \transstep \Delta_T $ and $ \left( \Delta_T, \Theta \right) \in \pointDis{\relationT} $.
					From $ \enc{P} \transstep \Delta_T $ we obtain $ P \transstep \Delta $ with $ \left( \Delta, \Delta_T \right) \in \pointDis{\relationT} $, because of soundness in \poc in Definition~\ref{def:POC}.
					By transitivity, we obtain $ \left( \Delta, \Theta \right) \in \pointDis{\relationT} $.
					Finally, $ \left( \Delta, \Theta \right) \in \pointDis{\relationEnc} $ follows from $ \left( \Delta, \Theta \right) \in \pointDis{\relationT} $, by the construction of $ \relationEnc $.
					\\
					Else, assume that $ P $ is a target term.
					Then $ \left( P, Q \right) \in \relationEnc $ implies $ \left( P, Q \right) \in \relationT $, because $ Q $ is a target term and $ \relationT = \restrictRelation{\relationEnc}{\processes{\target}} $.
					Since $ \relationT $ is a probabilistic bisimulation (see Definition~\ref{def:probabilisticBisimulation}), then $ Q \transstep \Theta $ implies $ P \transstep \Delta $ and $ \left( \Delta, \Theta \right) \in \pointDis{\relationT} $.
					Finally, $ \left( \Delta, \Theta \right) \in \pointDis{\relationEnc} $ follows from $ \left( \Delta, \Theta \right) \in \pointDis{\relationT} $, by the construction of $ \relationEnc $.
			\end{itemize}
		\item[only if ($ \longleftarrow $):] We assume that there is a relation $ \relationEnc $ such that $ \forall S \logdot \left( S, \enc{S} \right) \in \relationEnc $, $ \relationT = \restrictRelation{\relationEnc}{\processes{\target}} $, $ \forall S, T \logdot \left( S, T \right) \in \relationEnc \longrightarrow \left( \enc{S}, T \right) \in \relationT $, and $ \relationEnc $ is a preorder and a probabilistic bisimulation.
			We start with probabilistic operational correspondence.
			\begin{description}
				\item[Completeness:] Assume $ S \transstep \Delta_S $.
					Since $ \left( S, \enc{S} \right) \in \relationEnc $ and because $ \relationEnc $ is a probabilistic bisimulation, then $ \enc{S} \transstep \Delta_T $ and $ \left( \Delta_S, \Delta_T \right) \in \pointDis{\relationEnc} $.
					By $ \forall S, T \logdot \left( S, T \right) \in \relationEnc \longrightarrow \left( \enc{S}, T \right) \in \relationT $ and Lemma~\ref{lem:liftCondition}, then $ \left( \Delta_S, \Delta_T \right) \in \pointDis{\relationEnc} $ implies $ \left( \enc{\Delta_S}, \Delta_T \right) \in \pointDis{\relationT} $.
				\item[Soundness:] Assume $ \enc{S} \transstep \Delta_T $.
					Since $ \left( S, \enc{S} \right) \in \relationEnc $ and because $ \relationEnc $ is a probabilistic bisimulation, then $ S \transstep \Delta_S $ and $ \left( \Delta_S, \Delta_T \right) \in \pointDis{\relationEnc} $.
					By $ \forall S, T \logdot \left( S, T \right) \in \relationEnc \longrightarrow \left( \enc{S}, T \right) \in \relationT $ and Lemma~\ref{lem:liftCondition}, then $ \left( \Delta_S, \Delta_T \right) \in \pointDis{\relationEnc} $ implies $ \left( \enc{\Delta_S}, \Delta_T \right) \in \pointDis{\relationT} $.
			\end{description}
			By Definition~\ref{def:POC}, then $ \enc{\cdot} $ is probabilistically operationally corresponding \wrt $ \relationT $.
			Finally, since $ \relationEnc $ is a preorder and a probabilistic bisimulation and because of the Lemmata~\ref{lem:relationTIsPreorder} and \ref{lem:relationTIsBisimulation}, $ \relationT $ is a preorder and a probabilistic bisimulation. \qed
	\end{description}
\end{proof}

\begin{definition}[Strong Probabilistic Bisimulation]
	\label{def:strongProbabilisticBisimulation}
	A relation $ \relation $ is a \emph{strong probabilistic (reduction) bisimulation} if for each $ (P, Q) \in \relation $:
	\begin{itemize}
		\item $ P \step \Delta $ implies $ \exists \Theta \logdot Q \step \Theta \wedge (\Delta, \Theta) \in \pointDis{\relation} $
		\item $ Q \step \Theta $ implies $ \exists \Delta \logdot P \step \Delta \wedge (\Delta, \Theta) \in \pointDis{\relation} $
	\end{itemize}
	Two terms are \emph{strong probabilistic bisimilar} if a strong probabilistic bisimulation relates them.
\end{definition}

\begin{definition}[Strong Probabilistic Operational Correspondence]
	\label{def:strongPOC}
	An encoding $\enc{\cdot} : \processes{\source} \to \processes{\target}$ is \emph{strongly probabilistic operationally corresponding} (strong \poc) \wrt $\relationT \subseteq \processes{\target}^2 $ if it is:
	\begin{description}
		\item[\quad Strongly Probabilistic Complete:] $ $\\
		\hspace*{2em} $ \forall S, \Delta_S \logdot S \step \Delta_S \text{ implies } \left( \exists \Delta_T \logdot \enc{S} \step \Delta_T \wedge \left( \enc{\Delta_S}, \Delta_T \right) \in \pointDis{\relationT} \right) $
		\item[\quad Strongly Probabilistic Sound:] $ $\\
		\hspace*{2em} $ \forall S, \Delta_T \logdot \enc{S} \step \Delta_T \text{ implies } \left( \exists \Delta_S \logdot S \step \Delta_S \wedge \left( \enc{\Delta_S}, \Delta_T \right) \in \pointDis{\relationT} \right) $
	\end{description}
\end{definition}

Again, we adapt Lemma~\ref{lem:relationTIsCorrespondence} to strong probabilistic bisimulation.

\begin{lemma}
	\label{lem:relationTIsStrongBisimulation}
	If $ \relationEnc $ is a strong probabilistic bisimulation and $ \relationT = \restrictRelation{\relationEnc}{\processes{\target}} $ then $ \relationT $ is a strong probabilistic bisimulation.
\end{lemma}

\begin{proof}
	Assume that $ \relationEnc $ is a strong probabilistic bisimulation and $ \relationT = \restrictRelation{\relationEnc}{\processes{\target}} $.
	Moreover, assume $ \left( T_1, T_2 \right) \in \relationT $.
	Because of $ \relationT = \restrictRelation{\relationEnc}{\processes{\target}} $, then $ \left( T_1, T_2 \right) \in \relationEnc $.
	\begin{description}
		\item[Case (i):] Assume $ T_1 \step \Delta $.
			Since $ \relationEnc $ is a strong probabilistic bisimulation, then $ T_2 \step \Theta $ and $ \left( \Delta, \Theta \right) \in \pointDis{\relationEnc} $.
			With $ \relationT = \restrictRelation{\relationEnc}{\processes{\target}} $ and Lemma~\ref{lem:preservationReduction}, then $ \left( \Delta, \Theta \right) \in \pointDis{\relationT} $.
		\item[Case (ii):] Assume $ T_2 \step \Theta $.
			Since $ \relationEnc $ is a strong probabilistic bisimulation, then $ T_1 \step \Delta $ and $ \left( \Delta, \Theta \right) \in \pointDis{\relationEnc} $.
			With $ \relationT = \restrictRelation{\relationEnc}{\processes{\target}} $ and Lemma~\ref{lem:preservationReduction}, then $ \left( \Delta, \Theta \right) \in \pointDis{\relationT} $.
	\end{description}
	We conclude that $ \relationT $ is a probabilistic bisimulation.
\end{proof}

Then can show Theorem~\ref{thm:strongPOC}:
\begin{theorem}[Strong \poc]
	\label{thm:strongPOC}
	$ \enc{\cdot} $ is strongly probabilistically operationally corresponding \wrt a preorder $ \relationT \subseteq \processes{\target}^2 $ that is a strong probabilistic bisimulation iff \\ $ \exists \relationEnc \logdot \left( \forall S \logdot \left( S, \enc{S} \right) \in \relationEnc \right) \wedge \relationT = \restrictRelation{\relationEnc}{\processes{\target}} \wedge \left( \forall S, T \logdot \left( S, T \right) \in \relationEnc \longrightarrow \left( \enc{S}, T \right) \in \relationT \right) \wedge \relationEnc $ is a preorder and a strong probabilistic bisimulation.
\end{theorem}

\begin{proof}[Proof of Theorem~\ref{thm:strongPOC}]
	We prove the two directions of the result separately.
	\begin{description}
		\item[if ($ \longrightarrow $):] Assume that $ \enc{\cdot} $ is strongly probabilistically operationally corresponding \wrt a preorder $ \relationT \subseteq \processes{\target}^2 $ that is a strong probabilistic bisimulation.
			We construct $ \relationEnc $ from $ \relationT $ by adding $ \left( S, \enc{S} \right) $ for all source terms $ S $ and then building the reflexive and transitive closure.
			Accordingly, $ \forall S \logdot \left( S, \enc{S} \right) \in \relationEnc $ holds by construction.
			Since we did not add any pairs of only source terms, \ie no pairs of the form $ \left( S_1, S_2 \right) $ where both $ S_1 $ and $ S_2 $ are source terms, and since the only such pairs added by the reflexive and transitive closure are of the form $ \left( S, S \right) $, we have $ \relationT = \restrictRelation{\relationEnc}{\processes{\target}} $.
			Next we prove that $ \forall S, T \logdot \left( S, T \right) \in \relationEnc \longrightarrow \left( \enc{S}, T \right) \in \relationT $.
			Therefore, fix some $ S $ and $ T $ and assume $ \left( S, T \right) \in \relationEnc $.
			By the construction of $ \relationEnc $, $ \left( S, \enc{S} \right) \in \relationEnc $ and all pairs relating a source and a target term contain a source term and its literal translation or result from such a pair, $ \relationT $, and the transitive closure in the construction of $ \relationEnc $.
			Hence, $ T = \enc{S} $ or $ \left( \enc{S}, T \right) \in \relationT $.
			In the former case, $ \left( \enc{S}, T \right) \in \relationT $ follows from the reflexivity of $ \relationT $.
			The latter case directly provides $ \left( \enc{S}, T \right) \in \relationT $.
			That $ \relationEnc $ is a preorder directly follows from the construction of $ \relationEnc $, because we used the reflexive and transitive closure.
			As last condition we have to show that $ \relationEnc $ is strong a probabilistic bisimulation.
			By Definition~\ref{def:strongProbabilisticBisimulation}, for all $ \left( P, Q \right) \in \relationEnc $ and all $ P \step \Delta $ we have to find $ Q \step \Theta $ such that $ \left( \Delta, \Theta \right) \in \relationEnc $ and for all $ Q \step \Theta $ we have to find $ P \step \Delta $ such that $ \left( \Delta, \Theta \right) \in \pointDis{\relationEnc} $.
			By the construction of $ \relationEnc $, $ P $ may be a source or target term, but $ Q $ is a target term.
			\begin{itemize}
				\item Assume $ P \step \Delta $.
					If $ P $ is a source term, then we have $ \enc{P} \step \Delta' $ with $ \left( \Delta, \Delta' \right) \in \pointDis{\relationT} $, because of strong completeness in strong \poc in Definition~\ref{def:strongPOC}.
					Since $ Q $ is a target term and by the construction of $ \relationEnc $, $ \left( P, Q \right) \in \relationEnc $ implies $ \left( \enc{P}, Q \right) \in \relationEnc $.
					Because of $ \relationT = \restrictRelation{\relationEnc}{\processes{\target}} $, then $ \left( \enc{P}, Q \right) \in \relationT $.
					Then $ Q \step \Theta $ and $ \left( \Delta', \Theta \right) \in \pointDis{\relationT} $, because $ \relationT $ is a strong probabilistic bisimulation.
					With the transitivity of $ \relationT $ and thus $ \pointDis{\relationT} $, $ \left( \Delta, \Delta' \right) \in \pointDis{\relationT} $ and $ \left( \Delta', \Theta \right) \in \pointDis{\relationT} $ imply $ \left( \Delta, \Theta \right) \in \pointDis{\relationT} $.
					Finally, $ \left( \Delta, \Theta \right) \in \pointDis{\relationEnc} $ follows from $ \left( \Delta, \Theta \right) \in \pointDis{\relationT} $, by the construction of $ \relationEnc $.
					\\
					Else, assume that $ P $ is a target term.
					Then $ \left( P, Q \right) \in \relationEnc $ implies $ \left( P, Q \right) \in \relationT $, because $ Q $ is a target term and $ \relationT = \restrictRelation{\relationEnc}{\processes{\target}} $.
					Since $ \relationT $ is a strong probabilistic bisimulation (see Definition~\ref{def:strongProbabilisticBisimulation}), then $ P \step \Delta $ implies $ Q \step \Theta $ with $ \left( \Delta, \Theta \right) \in \pointDis{\relationT} $.
					Finally, $ \left( \Delta, \Theta \right) \in \pointDis{\relationEnc} $ follows from $ \left( \Delta, \Theta \right) \in \pointDis{\relationT} $, by the construction of $ \relationEnc $.
				\item Assume $ Q \step \Theta $.
					If $ P $ is a source term, then $ \left( \enc{P}, Q \right) \in \relationEnc $, by the construction of $ \relationEnc $.
					Because of $ \relationT = \restrictRelation{\relationEnc}{\processes{\target}} $, then $ \left( \enc{P}, Q \right) \in \relationT $.
					Since $ \relationT $ is a strong probabilistic bisimulation (see second case of Definition~\ref{def:strongProbabilisticBisimulation}), then $ Q \step \Theta $ implies $ \enc{P} \step \Delta_T $ and $ \left( \Delta_T, \Theta \right) \in \pointDis{\relationT} $.
					From $ \enc{P} \step \Delta_T $ we obtain $ P \step \Delta $ with $ \left( \Delta, \Delta_T \right) \in \pointDis{\relationT} $, because of strong soundness in strong \poc in Definition~\ref{def:strongPOC}.
					By transitivity, we obtain $ \left( \Delta, \Theta \right) \in \pointDis{\relationT} $.
					Finally, $ \left( \Delta, \Theta \right) \in \pointDis{\relationEnc} $ follows from $ \left( \Delta, \Theta \right) \in \pointDis{\relationT} $, by the construction of $ \relationEnc $.
					\\
					Else, assume that $ P $ is a target term.
					Then $ \left( P, Q \right) \in \relationEnc $ implies $ \left( P, Q \right) \in \relationT $, because $ Q $ is a target term and $ \relationT = \restrictRelation{\relationEnc}{\processes{\target}} $.
					Since $ \relationT $ is a strong probabilistic bisimulation (see Definition~\ref{def:strongProbabilisticBisimulation}), then $ Q \step \Theta $ implies $ P \step \Delta $ and $ \left( \Delta, \Theta \right) \in \pointDis{\relationT} $.
					Finally, $ \left( \Delta, \Theta \right) \in \pointDis{\relationEnc} $ follows from $ \left( \Delta, \Theta \right) \in \pointDis{\relationT} $, by the construction of $ \relationEnc $.
			\end{itemize}
		\item[only if ($ \longleftarrow $):] We assume that there is a relation $ \relationEnc $ such that $ \forall S \logdot \left( S, \enc{S} \right) \in \relationEnc $, $ \relationT = \restrictRelation{\relationEnc}{\processes{\target}} $, $ \forall S, T \logdot \left( S, T \right) \in \relationEnc \longrightarrow \left( \enc{S}, T \right) \in \relationT $, and $ \relationEnc $ is a preorder and a strong probabilistic bisimulation.
			We start with strong probabilistic operational correspondence.
			\begin{description}
				\item[Strong Completeness:] Assume $ S \step \Delta_S $.
					Since $ \left( S, \enc{S} \right) \in \relationEnc $ and because $ \relationEnc $ is a strong probabilistic bisimulation, then $ \enc{S} \step \Delta_T $ and $ \left( \Delta_S, \Delta_T \right) \in \pointDis{\relationEnc} $.
					By $ \forall S, T \logdot \left( S, T \right) \in \relationEnc \longrightarrow \left( \enc{S}, T \right) \in \relationT $ and Lemma~\ref{lem:liftCondition}, then $ \left( \Delta_S, \Delta_T \right) \in \pointDis{\relationEnc} $ implies $ \left( \enc{\Delta_S}, \Delta_T \right) \in \pointDis{\relationT} $.
				\item[Strong Soundness:] Assume $ \enc{S} \step \Delta_T $.
					Since $ \left( S, \enc{S} \right) \in \relationEnc $ and because $ \relationEnc $ is a strong probabilistic bisimulation, then $ S \step \Delta_S $ and $ \left( \Delta_S, \Delta_T \right) \in \pointDis{\relationEnc} $.
					By $ \forall S, T \logdot \left( S, T \right) \in \relationEnc \longrightarrow \left( \enc{S}, T \right) \in \relationT $ and Lemma~\ref{lem:liftCondition}, then $ \left( \Delta_S, \Delta_T \right) \in \pointDis{\relationEnc} $ implies $ \left( \enc{\Delta_S}, \Delta_T \right) \in \pointDis{\relationT} $.
			\end{description}
			By Definition~\ref{def:strongPOC}, then $ \enc{\cdot} $ is strongly probabilistically operationally corresponding \wrt $ \relationT $.
			Finally, since $ \relationEnc $ is a preorder and a strong probabilistic bisimulation and because of the Lemmata~\ref{lem:relationTIsPreorder} and \ref{lem:relationTIsStrongBisimulation}, $ \relationT $ is a preorder and a strong probabilistic bisimulation. \qed
	\end{description}
\end{proof}

Similar to Lemma~\ref{lem:successSensitivenessTargetRelation}, we show that $ \equiv $ is barb sensitive.

\begin{lemma}[$ \equiv $ is Barb Sensitive]
	\label{lem:barbSensitivenessTargetRelation}
	$ $\\
	If $ T_1 \equiv T_2 $ then $ \HasBarb{T_1}{n} \longleftrightarrow \HasBarb{T_2}{n} $ for all $ n \in \names \cup \out{\names} $.\\
	Moreover, if $ \Delta_1 \equiv \Delta_2 $ then $ \HasBarb{\Delta_1}{n} \longleftrightarrow \HasBarb{\Delta_2}{n} $ for all $ n \in \names \cup \out{\names} $.
\end{lemma}

\begin{proof}
	Fix some $ n \in \names \cup \out{\names} $.
	The proof is by induction on the definition of $ \equiv $.
	All cases are immediate.
	\begin{description}
		\item[$ \alpha $-Equivalence $ \equiv_{\alpha} $:] In this case $ T_1 \equiv_{\alpha} T_2 $.
			Since barbs via labelled steps on unrestricted names, we have $ \HasBarb{T_1}{n} $ iff $ \HasBarb{T_2}{n} $.
		\item[$ \para{P}{\nul} \equiv P $:] In this case $ T_1 = T_2 \mid \nul $.
			Since $ \nul $ does not contain any barbs, then $ \HasBarb{T_1}{n} $ iff $ \HasBarb{T_2}{n} $.
		\item[$ \para{P}{Q} \equiv \para{Q}{P} $:] In this case $ T_1 = P \mid Q $ and $ T_2 = Q \mid P $.
			Then $ \HasBarb{T_1}{\success} $ iff $ {\left( \HasBarb{P}{n} \vee \HasBarb{Q}{n} \right)} $ iff $ \HasBarb{T_2}{n} $.
		\item[$ \para{P}{\paraBrack{Q}{R}} \equiv \para{\paraBrack{P}{Q}}{R} $:] In this case $ T_1 = P \mid {\left( Q \mid R \right)} $ as well as $ T_2 = {\left( P \mid Q \right)} \mid R $.
			Thereby, $ \HasBarb{T_1}{\success} $ iff $ {\left( \HasBarb{P}{n} \vee \HasBarb{Q}{n} \vee \HasBarb{R}{n} \right)} $ iff $ \HasBarb{T_2}{n} $.
		\item[$ \probPiRes{x}{\nul} \equiv \nul $:] In this case $ T_1 = \probPiRes{x}{\nul} $ and $ T_2 = \nul $.
			Since $ \nul $ does not contain any barbs, then $ \NotHasBarb{T_1}{n} $ and $ \NotHasBarb{T_2}{n} $.
		\item[$ \probPiRes{xy}{P} \equiv \probPiRes{yx}{P} $:] Then $ T_1 = \probPiRes{xy}{P} $ and $ T_2 = \probPiRes{yx}{P} $.
			Since only unrestricted actions can be observed, then $ \HasBarb{T_1}{n} $ iff $ \HasBarb{P}{n} $ iff $ \HasBarb{T_2}{n} $.
		\item[$ \probPiRes{x}{\paraBrack{P}{Q}} \equiv \para{P}{\probPiRes{x}{Q}} $:] In this case $ T_1 = \probPiRes{x}{\left( P \mid Q \right)} $ and $ T_2 = P \mid \probPiRes{x}{Q} $, where $ x \notin \freeNames{P} $.
			Here $ x \notin \freeNames{P} $ ensures that $ \NotHasBarb{P}{x} $ and $ \NotHasBarb{P}{\out{x}} $.
			Then $ \HasBarb{T_1}{n} $ iff $ {\left( {\left( \HasBarb{P}{n} \vee \HasBarb{Q}{n} \right)} \wedge n \neq x \wedge n \neq \out{x} \right)} $ iff $ \HasBarb{T_2}{n} $.
		\item[$ \Delta_1 \equiv \Delta_2 $:] In this case there is a finite index set $ I $ such that $ \Delta_1 = \sum_{i \in I} p_i \pointDis{P_i} $, $ \Delta_2 = \sum_{i \in I} p_i \pointDis{Q_i} $, and $ P_i \equiv Q_i $ for all $ i \in I $.
			Because of $ P_i \equiv Q_i $, we have $ \HasBarb{P_i}{n} $ iff $ \HasBarb{Q_i}{n} $ for all $ i \in I $.
			Then $ \HasBarb{\Delta_1}{n} $ iff $ \HasBarb{\Delta_2}{n} $. \qed
	\end{description}
\end{proof}

Then we can show that the encoding $ \outerEncoding $ also respects barbs, \ie:
\begin{lemma}[Barb Sensitiveness, $ \outerEncoding $/$ \encPCCSPPi $]
	\label{lem:barbSensitiveness}
	$ $\\
	For every $ S $ and all $ n \in \names \cup \out{\names} $, $ \ReachBarb{S}{n} $ iff $ \ReachBarb{\OuterEncoding{S}}{\RenamingPCCSPPi{n}} $.
\end{lemma}

\begin{proof}[Proof of Lemma~\ref{lem:barbSensitiveness}]
	By Definition~\ref{def:encCCSPi}, \ie since the encoding function does not introduce new free names and because of the rigorous use of the renaming policy $ \renamingPCCSPPi $, $ \HasBarb{S^*}{n} $ iff $\HasBarb{\OuterEncoding{S^*}}{\RenamingPCCSPPi{n}} $ for all $ S^* $.
	Then also $ \HasBarb{\Delta_S^*}{n} $ iff $\HasBarb{\OuterEncoding{\Delta_S^*}}{\RenamingPCCSPPi{n}} $ for all distributions $ \Delta_S^* $.
	\begin{description}
		\item If $ \ReachBarb{S}{n} $, then $ S \transstep \Delta_S $ and $ \HasBarb{\Delta_S}{n} $.
			By Lemma~\ref{lem:completeness}, then $ \OuterEncoding{S} \transstep \Delta_T $ and $ \OuterEncoding{\Delta_S} \equiv \Delta_T $.
			By Definition~\ref{def:encCCSPi}, $ \HasBarb{\Delta_S}{n} $ implies $ \HasBarb{\OuterEncoding{\Delta_S}}{\RenamingPCCSPPi{n}} $.
			By Lemma~\ref{lem:barbSensitivenessTargetRelation}, then $ \OuterEncoding{\Delta_S} \equiv \Delta_T $ implies $ \HasBarb{\Delta_T}{\RenamingPCCSPPi{n}} $.
			Finally, $ \OuterEncoding{S} \transstep \Delta_T $ and $ \HasBarb{\Delta_T}{\RenamingPCCSPPi{n}} $ imply $ \ReachBarb{\OuterEncoding{S}}{\RenamingPCCSPPi{n}} $.
		\item If $ \ReachBarb{\OuterEncoding{S}}{\RenamingPCCSPPi{n}} $, then $ \OuterEncoding{S} \transstep \Delta_T $ and $ \HasBarb{\Delta_T}{\RenamingPCCSPPi{n}} $.
			By Lemma~\ref{lem:soundness}, then $ S \transstep \Delta_S' $, $ \Delta_T \transstep \Delta_T' $, and $ \OuterEncoding{\Delta_S'} \equiv \Delta_T' $.
			Because of $ \Delta_T \transstep \Delta_T' $, $ \HasBarb{\Delta_T}{\RenamingPCCSPPi{n}} $ implies $ \HasBarb{\Delta_T'}{\RenamingPCCSPPi{n}} $.
			By Lemma~\ref{lem:barbSensitivenessTargetRelation}, $ \HasBarb{\Delta_T'}{\RenamingPCCSPPi{n}} $ and $ \OuterEncoding{\Delta_S'} \equiv \Delta_T' $ imply $ \HasBarb{\OuterEncoding{\Delta_S'}}{\RenamingPCCSPPi{n}} $.
			By Definition~\ref{def:encCCSPi}, then $ \HasBarb{\Delta_S'}{n} $.
			Finally, $ S \transstep \Delta_S' $ and $ \HasBarb{\Delta_S'}{n} $ imply $ \ReachBarb{S}{n} $. \qed
	\end{description}
\end{proof}
\end{adjustwidth}

\bibliographystyle{splncs04}
\bibliography{lit}

\begin{thebibliography}{1}
\providecommand{\url}[1]{\texttt{#1}}
\providecommand{\urlprefix}{URL }
\providecommand{\doi}[1]{https://doi.org/#1}

\bibitem{Deng07}
Deng, Y., Du, W.: {Probabilistic Barbed Congruence}. Electronic Notes in
  Theoretical Computer Science  \textbf{190}(3),  185--203 (2007).
  \doi{10.1016/j.entcs.2007.07.011}

\bibitem{FournetGonthier96}
Fournet, C., Gonthier, G.: {The Reflexive CHAM and the Join-Calculus}. In:
  Proc. of POPL. pp. 372--385. ACM (1996). \doi{10.1145/237721.237805}

\bibitem{Gorla10}
Gorla, D.: {Towards a Unified Approach to Encodability and Separation Results
  for Process Calculi}. Information and Computation  \textbf{208}(9),
  1031--1053 (2010). \doi{10.1016/j.ic.2010.05.002}

\bibitem{Milner89}
Milner, R.: {Communication and Concurrency}. Prentice-Hall, Inc. (1989)

\bibitem{Milner99}
Milner, R.: {Communicating and mobile systems - the Pi-calculus}. Cambridge
  University Press (1999)

\bibitem{PetersGlabbeek15}
Peters, K., van Glabbeek, R.: {Analysing and Comparing Encodability Criteria}.
  In: Proc. of EXPRESS/SOS. EPTCS, vol.~190, pp. 46--60 (2015).
  \doi{10.4204/EPTCS.190.4}

\bibitem{Sangiorgi96}
Sangiorgi, D.: {$\pi$-Calculus, internal mobility, and agent-passing calculi}.
  Theoretical Computer Science  \textbf{167}(1),  235--274 (1996).
  \doi{https://doi.org/10.1016/0304-3975(96)00075-8}

\bibitem{schmitt23}
Schmitt, A., Peters, K.: {Probabilistic Operational Correspondence}. In: Proc.
  of CONCUR. Springer (2023)

\bibitem{Varacca07}
Varacca, D., Yoshida, N.: {Probabilistic pi-Calculus and Event Structures}.
  Electronic Notes in Theoretical Computer Science  \textbf{190}(3),  147--166
  (2007). \doi{10.1016/j.entcs.2007.07.009}

\end{thebibliography}

\end{document}